%% file: fomon.tex
\documentclass[12pt]{article}
\usepackage{hyperref}
\usepackage{a4wide}
\usepackage{amsfonts}
\usepackage{amssymb}
\usepackage{amsmath}
\usepackage{mathtools}
\usepackage{rotating}
\usepackage{comment}
\usepackage{tikz}
\usepackage{tweaklist}

\input{defs.tex}

\begin{document}
\title{First-Order Quantifiers \\
  and the Syntactic Monoid \\
  of Height Fragments
  of Picture Languages\\
}
\author{Oliver Matz\\
  {\small
  Institut f\"ur Informatik,
  Universit\"at Kiel, Germany}\\
  {\small
  \texttt{matz@ti.informatik.uni-kiel.de}}}

\maketitle

\begin{abstract}
  We investigate the expressive power of first-order quantifications in the
  context of monadic second-order logic over pictures.  We show that
  $k+1$ set quantifier alternations
  allow to define a picture language that cannot be defined using $k$ set
  quantifier alternations preceded by arbitrarily many first-order quantifier
  alternations.

  The approach uses, for a given picture language $L$ and an integer $m\geq 1$,
  the \emph{height-$m$ fragment of $L$}, which is defined as the word language
  obtained by considering each picture $p$ of height $m$ in $L$ as a word, where
  the letters of that word are the columns of $p$.

  A key idea is to measure the complexity of a regular word language by 
  the group complexity of its syntactic monoid.
  Given a picture language $L$, such a word language measure may be applied to
  each of its height fragments, so that the complexity of the picture language
  is a function that maps each $m$ to the complexity of the height-$m$
  fragment of $L$.
  %
  The asymptotic growth rate of that function may be bounded based on the
  structure of a monadic second-order formula that defines $L$.

  The core argument for that lower bound proof is based on Straubing's
  algebraic characterization of the effect of first-order quantifiers on the
  syntactic monoid of word languages by means of Rhodes' and Tilson's block
  product.

  \paragraph{Keywords:}
  Picture languages, monadic second-order logic, quantifier alternation,
  syntactic monoid, group complexity
\end{abstract}

\pagebreak
\tableofcontents{}
\pagebreak

\section{Introduction}

In monadic second-order logic (MSO) over finite structures, formulas use
first-order quantifications (ranging over elements of the universe) as well as
set quantifications (ranging over sets thereof).  
\cite{MT97, Schw97, Matz-Diss, MST98, Matz02} investigate the effect of the alternation of
existential and universal set quantification and show that the depth of this
alternation cannot be bounded without loss of expressive power.  The proofs are
done for a specific class of structures, namely pictures.

The upper bound proofs in those papers show that very little use of set
quantification is needed.  The quantifiers that actually alternate
are all first-order.  Set quantification is needed for two purposes: The
outermost set quantification establishes a specific, uniquely determined
coloring, and the innermost
set quantification is needed only to replace the horizontal ordering, which itself
is not present in the logic.

Thus all the formulas constructed in the upper bound proofs can be written
in prenex normal form with a quantifier prefix of the form
\begin{equation*}
  \label{eq:twoblocks}
  \EX^* \{\ex,  \fa\}^* \EX^* \{\ex, \fa\}^*,
\end{equation*}
where $\EX^*$ denotes a block of existential set quantifiers, and $\fa$ and
$\ex$ denote universal (or existential, respectively) first-order
quantifiers.  The question whether every MSO formula can be written in this
form remains open, see Problem~\ref{prob:mso}.

This motivates the interest in the power of first-order quantification in the
context of MSO. It has been
studied in \cite{AFS00, JM01} in the context of graphs.  
In \cite{AFS00}, the authors suggest the
\emph{closed MSO alternation hierarchy}, which is
coarser and more robust than the ordinary MSO alternation hierarchy of \cite{MT97}
because it allows to intersperse first-order quantifiers ``for free'' between
set quantifiers.  The authors ask whether this hierarchy is strict---a
question that is still open.

In \cite{JM01}, the authors develop a technique to
infer new separation results dealing with the first-order closure.
Specifically, they show the following:
\begin{thThe}[\cite{JM01}]
  Let $V,W \subseteq \bigset{\EX, \FA, \ex,\fa}^*$.  Let $S$ be a graph property
  definable with a quantifier prefix of a form in $V$ but not with one of a
  form in $W$. Then there is another property definable with a quantifier prefix of a
  form in $\EX\,\fa\fa\, V$ but not with one of a form in $\{\ex,\fa\}^*\, W$.
\end{thThe}
The authors of \cite{JM01} apply that theorem
to show the following corollary (previously shown directly
in \cite{AFS00}).
\begin{thCor}\label{cor:AFS00}
  There exists a graph property definable by a prenex normal form 
  of type 
  \(
  \EX^*\{\ex,\fa\}^*\EX^*\{\ex,\fa\}^*
  \)
  but not with one of type 
  \(
  \{\ex,\fa\}^*\EX^*\{\ex,\fa\}^*.
  \)
\end{thCor}

In this paper, we focus on pictures (as opposed to arbitrary finite graphs).
We show that the above corollary is true for pictures,
too, thereby giving yet another proof for it (see
Corollary~\ref{cor:AFS00-gen} for the case $k=1$).
Besides we consider (as in \cite{Matz-Diss}) formulas that have a quantifier
prefix of the form
\[
\{\ex, \fa\}^* \{\EX, \FA\}^* \{\ex, \fa\}^*,
\]
where the set quantifier block in the middle contains only $k$ alternations,
and compare their expressive power to formulas with a quantifier prefix of the form
\[
\{\EX, \FA\}^* \{\ex, \fa\}^*,
\]
where the set quantifier block contains only $k+1$ alternations.  The main
result of this paper (Corollary~\ref{cor:main-sep}) is that there is a
formula of the latter kind that is not equivalent to any of the former kind.
Once again, the formula constructed in the proof does not actually use $k+1$
set quantifier alternations; it has only two set quantifier blocks---the
$k+1$ alternations stem from the first-order quantifier blocks in between these.

\input{tikz-hierarchy}

The lower bound proof is based on the block product introduced in
\cite{RhoTil89}, which, by \cite{Straub94}, allows to characterize the effect
of first-order quantifiers on the syntactic monoid of word languages, see
Lemma~\ref{lem:semexists}.  The application to pictures (as opposed to
words) follows the same approach as \cite{MT97, MST98, Matz-Diss, Matz02}: The
common essential idea to show that a picture language $L$ 
cannot be defined by a formula class $\cF$ is the following: 
We consider the family $(\fix Lm)_{m \geq 1}$, where $\fix Lm$ contains the
pictures of height $m$.  That so-called \emph{height-$m$ fragment} $\fix Lm$ may be
regarded as a word language over the $m$-fold Cartesian product of the
alphabet.  Then we show that, for a sufficiently large $m$,
the complexity of that height-$m$ fragment (wrt.\ some suitable
complexity measure of word languages) is too high, so that $L$ cannot be defined by a
$\cF$-formula.

Typical complexity measures used in \cite{MT97, MST98, Matz-Diss} are firstly
the number of states needed for a recognizing non-deterministic finite
automaton and, secondly, the length of the shortest word of a (unary) word
language.  This paper is the first one in which the used complexity measure is
the group complexity of the syntactic monoid.

\section{Basic Notions}

\subsection{Pictures}

Let $\Gamma$ be a finite alphabet. A \emph{non-empty picture} of size $(m,n)$
over $\Gamma$ (where $m,n\geq 1$) is an $m{\times}n$-matrix over $\Gamma$.  If
$p$ is a non-empty picture of size $(m,n)$, we denote the length $n$ (i.e.,
the number of columns) by $\length p$, the height $m$ (i.e., the number of
rows) by $\height p$, and the domain $\ganzrec mn$ by $\dom p$.  The component
at position $(i,j) \in \dom p$ is denoted $p\pos ij$.

The set of non-empty pictures over $\Gamma$ is denoted $\Gammapp$.  A set of
non-empty pictures is called a \emph{picture language}.

If $p$ and $q$ are non-empty pictures with $\height p = \height q$, 
then the non-empty picture that results by
appending $q$ to the right of $p$ is denoted $p \pcol q$.
This partial operation is called the \emph{column concatenation}.

Picture languages must not contain the empty picture.  Nevertheless, when we
assemble picture languages by column concatenation, it is often convenient to
have a neutral element.  That is why we consider a special, distinct
\emph{empty picture} which we denote $\epsilon$ and for which height, width, domain,
and size are not defined.  For every (empty or non-empty) picture $p$, we
define $p \pcol \epsilon = \epsilon \pcol p = p$.

Column concatenation is lifted to sets of pictures as
usual.
For every set $L$ of pictures, the iterated column concatenation is defined by
$L^0 = \{\epsilon\}$ and $L^{i+1} = L^i \pcol L$ and every $i \geq 0$.

The set of all non-empty pictures of height $m$ over alphabet $\Gamma$ is
denoted $\Gamma^{m,+}$, and $\Gamma^{m,*}$ abbreviates $\Gamma^{m,+} \cup \{\epsilon\}$.

Let $p$ be a non-empty picture over $\Gamma$ and $m = \height p$.  We frequently
consider each column of $p$ as a letter of the new alphabet $\Gamma^m$.  This
way, we identify every non-empty picture $p$ with a word of length $\length p$ over
alphabet $\Gamma^m$.
For a set of pictures $L$ and $m\geq 1$ we define the \emph{height-$m$
  fragment} (denoted $\fix Lm$) as the set of these words over alphabet $\Gamma^m$.

\subsection{Pictures over Attributes Alphabets}

While in general the nature of an alphabet $\Gamma$ is indifferent, it will be
technically convenient to have certain notions for the case that $\Gamma$ is of the form
$\schema I$ for a finite set of so-called \emph{attributes}.  That means that each
letter $a\in\Gamma$ is a mapping $I \rightarrow \{0,1\}$.


If $a\in \schema{I}$ and $J\subseteq I$, then $\restr{J}{a} = a \restrict J$ is the
restriction of $a$ to a $J$-indexed family.  The mapping $\textsf{restr}_{J}$
is an alphabet projection from $\schema I \rightarrow \schema J$, which is
lifted to pictures and picture languages the usual way.

The alphabet projection
$\Exset{J} : \schema I \rightarrow\schema{I\backslash J}$ is defined by
$\Exset{J} = \Restr{I\backslash J}$.

Furthermore, we define for every $\mu\in I$ the mapping 
\[
\textsf{pr}_{\mu} : \schema I \rightarrow \{0,1\},\quad a \mapsto a(\mu).
\]
It is an alphabet projection, too, and it is lifted to pictures and picture
languages the usual way.

Typically, each attribute corresponds to a free variable in a formula (see
next section).  With regard to sentences (i.e., formulas without free
variables), it is therefore consequent to allow also the empty attribute set,
so by convention, $\schema\emptyset$ is some fixed singleton alphabet,
and $\Restr\emptyset$ denotes the alphabet projection to that singleton
alphabet.

\subsection{Monadic Second-Order Formulas}

We describe our conventions for formulas.  We will be concerned with a
fixed signature with two binary successor predicates $S_1$, $S_2$ and
with the specific class of structures associated to non-empty pictures.

Let $J, K$ be two disjoint sets of attributes, which we use as
indices of variables.  We use first-order variables $x_\nu$ with $\nu\in K$
and set variables $X_\mu$ with $\mu \in J$.  Atomic formulas are of the
form $\X{\mu}{x_\nu}$ (for $\mu\in J$ and $\nu \in K$), or $S_1x_\mu x_\nu$, or
$S_2x_\mu x_\nu$, or $x_\mu=x_\nu$ (for $\mu,\nu\in K$).
Formulas are assembled in the usual
way using the boolean connectives as well as first-order
quantification ($\exists x_\nu \phi$ or $\forall x_\nu \phi$) 
and set quantification ($\exists X_\mu \phi$ or $\forall X_\mu \phi$).

\subsection{Pictures as Models}

Let $J, K$ be two disjoint attribute sets. Set $I = J \cup K$.
%
Let $\Unique IK$ be the set of those non-empty pictures $p$ over alphabet $\schema I$
such that for all $\nu\in K$ there is exactly one position $(i,j)\in\dom p$
with $p\pos ij(\nu) = 1$.

Let $p\in\Unique IK$, say with size $(m,n)$.  Let $\phi$ be a
formula with free set variables in $\{ X_\mu \mid \mu \in J\}$ and free
first-order variables in $\{ x_\nu \mid \nu \in K\}$.  To $p$, we associate
the grid structure with universe $\ganzrec mn$ 
and an assignment $(X_\mu^p)_{\mu\in J}, (x_\nu^p)_{\nu\in K}$ to the
free variables in the following way:
\begin{itemize}
\item $X_\mu^p = \{ (i,j) \in \dom p \mid p\pos ij(\mu) = 1 \}$,
\item $x_\nu^p$ is the unique $(i,j) \in \dom p$ with $p\pos ij(\nu) = 1$.
\end{itemize}
We write 
\[
p \models \phi
\]
iff this assignment makes $\phi$ true in the
structure with universe $\ganzrec mn$, where the predicates $S_1$ and $S_2$
are interpreted as the vertical and horizontal successor relation, respectively, i.e.,
$S_2x_1x_2$ asserts that $x_2$ is the horizontal successor of $x_1$.
That means, $p\models S_2x_1x_2$ iff there exist $\pos ij, \pos i{j+1} \in
\dom p$ such that $x_1^p = \pos ij$ and $x_2^p = \pos i{j+1}$.

Another notation convention will be convenient.  Let $\phi$ and $p$ be as
above.  Suppose that $\nu_1,\ldots, \nu_n\in K$ are attributes of first-order
variables. If $a_1,\ldots,a_n\in \dom p$, then, by abuse of notation, we write
\[
p, \subst{a_1}{x_{\nu_1}} \ldots \subst{a_n}{x_{\nu_n}} \models \phi
\]
iff $\phi$ is made true in the structure from above, where the assignment
for the attributes $\{\nu_1,\ldots, \nu_n\}$ is provided by setting
$x_{\nu_j}^p = a_j$ for every $j\in\vonbis1n$.

Let $\Mod{I,K}{\phi} = \{ p \in \Unique IK \mid p \models \phi \}$.
We write $\Mod{}{\phi}$ rather than $\Mod{I,K}{\phi}$ if $I,K$
are clear from the context; typically $I$ (or $K$) is the set of those attributes that
may appear as indices of free variables (or free first-order
variables, respectively) of $\phi$, or any superset thereof.
Indeed we have the following remark, which shows that adding an element to
the attribute set does not make that much of a difference:
\begin{thRem}\label{rem:restr}
  Let $I$ be an attribute set, $K\subseteq I$, and $\mu \in I \backslash K$.
  If $\phi$ is a formula with indexes of free variables in $I\backslash \{\mu\}$,
  then 
  \[
  \Mod{I, K}{\phi} =
  \exsetinv{\{\mu\}}{\Mod{I\backslash \{\mu\},K}{\phi}}.
  \]
\end{thRem}

%

The concept of existential set quantification is captured by the
alphabet projection on picture languages in the following sense, motivating the
notation $\Exset{\{\mu\}}$ for the alphabet projection.
\begin{thRem}\label{rem:exists-set}
  Let $I$ be an attribute set, $K\subseteq I$, and $\mu \in I \backslash K$.
  If $\phi$ is a formula with indexes of free variables in $I$, then
  \[
  \Mod{I\backslash\{\mu\},K}{\exists X_\mu \phi} = 
  \exset{\{\mu\}}{\Mod{I,K}{\phi}}.
  \]
\end{thRem}

Two formulas $\phi, \psi$ are \emph{equivalent} iff $\Mod{}{\phi} =
\Mod{}{\psi}$.
Note that our notion of equivalence implicitly refers to the class
of pictures and is thus coarser than logical equivalence.

A formula $\phi$ \emph{defines} a picture language $L$ over alphabet $\schema
J$ if $\phi$ has no free first-order variables, the free set variables of
$\phi$ are among $(X_\mu)_{\mu \in J}$, and $\Mod{J,\emptyset}{\phi} = L$.

We use the following convention for variable substitution.  Let $\phi$ be a
formula.  Let $X_1,\ldots, X_m$ (and $x_1,\ldots, x_n$) be set
variables (or first-order variables, respectively).  Note that we do not
require that all free variables of $\phi$ are among these, but typically,
this is the case.

We may write $\phi(X_1,\ldots, X_m, x_1,\ldots, x_n)$ instead
of $\phi$ in order to pick these variables for a later substitution. 
If we later write $\phi(X'_1,\ldots, X'_m, x'_1,\ldots, x'_n)$, 
for other variables $X'_1,\ldots, X'_m, x'_1,\ldots, x'_n$, 
then this
denotes the formula that results from $\phi$ by replacing each indicated
variable from the first variable tuple by the respective variable from the
latter.

For example, if we introduce the formula $\phi$ as $\phi(x_1,x_2)$, then by
$\phi(x_2,x_1)$ we mean the formula that results from $\phi$ by exchanging
the occurrences of $x_1$ and $x_2$.

\begin{Exp}\label{exp:ftop}
  The first-order formula $\ftop(x) := \neg\exists y (S_1 yx)$ asserts for a
  position $x$ that it is in the top row.
  Similarly, $\fleft(x) := \neg\exists y(S_2yx)$ and $\fright(x) :=
  \neg\exists y(S_2xy)$ assert that $x$ is in the leftmost (or rightmost,
  respectively) column.
\end{Exp}

\begin{Exp}\label{exp:leq}
  Let 
  \[
   \psi = \forall x_1\forall x_2 
     ((S_2x_1x_2 \wedge \X\cld x_1) \rightarrow \X\cld x_2).
  \]
  Then $X_\cld$ is the only free set variable of $\psi$.  The formula $\psi$
  asserts that $X_\cld$ is closed under horizontal successors.  In other
  words, for a non-empty picture $p$ over alphabet $\schema{\{\cld\}}$ we have that $p
  \models \psi$ iff every row of the picture $\attr p\cld$ is in $0^*1^*$.
  Let 
  \[
  \phi = \forall X_{\cld} (\X\cld x \wedge \psi \rightarrow \X\cld x').
  \]
  Then $\psi$ asserts that position $x'$ is right to $x$.  More precisely:
  for a non-empty picture $p$ and two positions $\pos ij, \pos{i'}{j'} \in \dom p$
  we have $p,\subst{\pos ij}{x}\subst{\pos {i'}{j'}}{x'} \models \phi$ iff $i=i'$ and $j\leq j'$.

  Consider the formula $\fright(x)$ from the preceding example.
  Let 
  \[
  \phi' = 
    \exists X_{\cld} \Big(
    \begin{array}[t]{@{}l}
      \X\cld x \wedge 
      \forall x_1 \big(\fright(x_1) \wedge x_1{\not=}x' 
      \rightarrow \neg\X\cld x_1\big) \wedge {}\\
      \forall x_1x_2 \big( 
      S_2x_1x_2 \wedge x_1{\not=}x'
                 \wedge \X\cld x_1 \rightarrow \X\cld x_2\big)\Big). \\
  \end{array}
  \]
  Then $\phi'$ is equivalent to $\phi$.
\end{Exp}

The formula $\phi$ from Example~\ref{exp:leq} 
will be abbreviated as $x\leq_2 x'$
and will be needed later.
The above is a standard example of how to use set quantification to express
the horizontal ordering, which we do not have in the signature.

\subsection{Quantifier Alternation Classes}
\label{sec:quantifier-alternation-classes}

In this section we define the formula classes that are characterized by the
structure of their prenex normal form wrt.\ blocks of existential or
universal set quantifiers or first-order quantifiers.

A \emph{first-order formula} is a formula that does not make use of set
quantification.  The class of first-order formulas is denoted $\FO$.
For a class $\cF$ of formulas, let $\co{\cF}$ be the class of formulas
$\neg \phi$ with $\phi \in \cF$.

Let $\cF$ be a class of formulas.  The 
\begin{enumerate}
\item \emph{boolean closure} of $\cF$, denoted
  $\clB{\cF}$,
\item \emph{existential first-order closure} of $\cF$, denoted
  $\FOSigma{1}{\cF}$,
\item \emph{existential monadic closure} of $\cF$, denoted
  $\SOSigma{1}{\cF}$,
\item \emph{first-order closure} of $\cF$, denoted $\clFO{\cF}$,
\end{enumerate}
respectively, are defined as the smallest superclass of $\cF$ that is
closed under
\begin{enumerate}
\item boolean combinations,
\item existential first-order quantifications and positive boolean
  combinations,
\item existential set quantifications and positive boolean
  combinations,
\item first-order quantifications and boolean combinations,
\end{enumerate}
respectively.
\pagebreak[1]

We define $\SOSigma0{\cF} = \cF$ and
$\SOSigma{k+1}{\cF} = \SOSigma1{\clB{\SOSigma k{\cF}}}$ for every
$k\geq 0$.  Let $\SOPi k{\cF} = \co{\SOSigma k{\co{\cF}}}$ for every
$k$.
We write $\sSigma k{}$ and $\sPi k{}$ instead of $\SOSigma k{\FO}$
and $\SOPi k{\FO}$, respectively.

The formula classes $\FOSigma k{\cF}$ and $\FOPi k{\cF}$ are defined analogously
but for first-order rather than set quantification.

In the sequence, every formula that is equivalent to a formula in $\cF$ will
be called an $\cF$-formula, too.

$\nDelta k$ is the class of formulas that are both a
$\sSigma k{}$-formula and a $\sPi k{}$-formula.

Some of the quantifier alternation classes can be characterized very
succinctly by giving regular expressions over the alphabet
$\{\EX,\FA,\ex,\fa\}$ to describe the quantification structure of their formulas
in prenex normal form, as we did in the introduction.  For example, $\sSigma
3{}$ corresponds to $\EX^*\FA^*\EX^*\{\ex, \fa\}^*$, and $\FOPi 2{\sPi 2{}}$
corresponds to $\fa^* \ex^* \FA^*\EX^* \{\ex,\fa\}^*$.

The situation is more difficult for classes that involve the boolean closure,
the first-order closure, or $\sDelta k$, as for such a class it is not
possible to give a corresponding expression.  For example, $\{\fa,\ex\}^*
\EX^* \{\ex,\fa\}^*$ corresponds to what is called the \emph{positive first-order
closure of $\sSigma 1$} in \cite{AFS00}, and that class is between $\clFO{\sDelta 1}$ and
$\clFO{\sSigma1{}}$.

\begin{Exp}\label{exp:leqdelta}
  In Example~\ref{exp:leq} we have seen that $x \leq_2 x'$ is a
  $\sSigma1$-formula and also a $\sPi1$-formula.  
  Hence it is a $\nDelta1$-formula.
\end{Exp}
Every boolean combination
of $\nDelta1{}$-formulas is a $\nDelta1{}$-formula.
By a standard argument, first-order quantification may be expressed by a
suitable set quantification in the following sense: 
\begin{thRem}\label{rem:set-fo}
  Let $\cF\supseteq \FO$ is a class of formulas closed under conjunction and
  $\phi \in \cF$.  Then there is a formula $\phi'\in \cF$ such that 
  $\exists x_\mu \phi$ is equivalent to $\exists X_\mu \phi'$.
\end{thRem}

By the standard calculation rules of predicate logic and by the above remark,
we have the following.
\begin{thRem}\label{rem:calc}
  Let $k\geq 1$ and $\phi$ be a $\FOPi{k}{\nDelta1{}}$-formula.  If $k$ is
  odd, then $\phi$ is a $\FOPi{k-1}{\sPi1{}}$-formula.
  If $k$ is even, then $\phi$ is a $\FOPi{k-1}{\sSigma1{}}$-formula.
  Hence in any case $\phi$ is a $\sPi k{}$-formula.
\end{thRem}

\subsection{Syntactic Congruence and Syntactic Monoid}
\label{subs:syncong}

Let $L$ be a word language over alphabet $\Gamma$.
The \emph{syntactic congruence} $\equiv_L$ is defined as follows:
For two words $x, y\in\Gamma^*$ we define $x\equiv_L y$ iff for all
$u,v\in\Gamma^*$ we have $uxv\in L \Leftrightarrow uyv\in L$.
The \emph{syntactic monoid} $\Mon L$ of $L$ is the quotient of $\Gamma^*$ by
this congruence, i.e., $\Mon L$ consists of all congruence classes of the
syntactic congruence of $L$. 

A comprehensive introduction to the concept of syntactic monoids of word
languages can be found in many text books, see e.g. \cite{Pin86} (especially
Section~2.4.) or \cite{Straub94} (especially Theorem~V.1.3).

A \emph{group contained in} a semigroup $S$ is a subsemigroup of $S$ that is a
group.  If $G$ is a group contained in a monoid $M$, the neutral element of $G$ may
or may not be the same as that of $M$.
A \emph{submonoid} of a monoid $M$ is a subsemigroup of $M$ that contains the
neutral element of $M$.

Let $S, T$ be semigroups.  Then $S$ \emph{divides} $T$ (written $S\prec T$)
iff $S$ is a homomorphic image of a subsemigroup of $T$. Then $\prec$ is
transitive.

\subsection{Pseudovarieties and Group Complexity}
\label{sec:group-complexity}
The group complexity was introduced in \cite{KroRho65} and assigns a
non-negative integer $\group S$ to every semigroup $S$.  We briefly introduce
the related notions, see e.g.\ \cite{Eil76, rhodes2009q}.


A class $\pV$ of finite semigroups is a \emph{pseudovariety} (see
\cite{rhodes2009q}, Definition~1.2.30)  if the following
properties hold:
\begin{itemize}
\item $\pV$ contains a one-element semigroup;
\item For all semigroups $S_1, S_2 \in \pV$ we have $S_1 \times S_2 \in \pV$,
  where $S_1 \times S_2$ is the direct product;
\item For all semigroups $S,T$ with $S\prec T$ and $T\in\pV$ we have $S\in\pV$.
\end{itemize}

The following are pseudovarieties:
\begin{itemize}
\item the class $\pG$ of finite groups, and
\item the class $\pA$ of finite aperiodic semigroups.
\end{itemize}

\begin{thDef}\label{def:left-action}
Let $S, T$ be finite semigroups, and let us write $S$ additively.
A \emph{left action} of $T$ on $S$ is a map $T\times S \rightarrow S$, where
the image of $(t,s)$ is denoted $ts$, that satisfies the following two
properties:
\begin{itemize}
\item $t(s+s') = ts + ts'$ for every $t\in T$ and every $s,s'\in S$,
\item $(tt')s = t(t's)$ for every $t,t'\in T$ and every $s\in S$.
\end{itemize}
\end{thDef}

For two semigroups $S$ and $T$, the \emph{semidirect product} $S\rtimes T$
wrt.\ a given left action is defined as the set $S\times T$ with
multiplication given by
\[
(s,t)(s',t') = (s+ts', tt').
\]


For two pseudovarieties $\pV$ and $\pW$, the \emph{semidirect product}
 $\pV * \pW$ is
the pseudovariety generated by the semigroups of the form $V \rtimes W$ with
$V\in\pV$ and $W\in\pW$.

We define (cf.\ \cite{rhodes2009q}, Definition~4.3.10):
\[
\begin{array}{rcl}
  \pC_0 & = & \pA,\\
  \pC_{n+1} & = & \pA * (\pG * \pC_n)\\
            & = &  \pC_n * (\pA * \pG) , \mbox{ for every } n\geq 0.
\end{array}
\]
The famous ``prime decomposition theorem'' of \cite{KroRho65} asserts that
every finite semigroup is in $\bigcup_{n\geq 1} \pC_n$.  Thus for every finite
semigroup
$S$, the \emph{group complexity} $\group S := \min\{ n \mid S \in \pC_n\}$
is well defined.

For a regular word language $L$, the \emph{group complexity of $L$} is defined
as $\group L = \group{\Mon L}$.


\subsection{Height Fragment Technique and Results}

Let $\cP$ be a class of picture languages over a fixed alphabet $\Gamma$
and let $\alpha$ be a function that assigns to every regular word language an
element from $\NN$.
We say that
\begin{itemize}
\item $\cP$ is \emph{at most $k$-fold exponential wrt.\ $\alpha$} if for every
  $L\in \cP$ we have that $\alpha(L[m])$ is at most $k$-fold exponential in $m$.
\item $\cP$ is \emph{at least $k$-fold exponential wrt.\ $\alpha$} if there exists
  $L\in \cP$ such that $\alpha(L[m])$ is at least $k$-fold exponential in $m$.
\end{itemize}

The following remark is immediate.  Its contraposition provides a technique to separate two
classes of picture languages, given a complexity measure $\alpha$ as above.
\begin{thRem}[Height Fragment Technique]\label{rem:asymptotic}
  Let $\cP,\cP'$ be picture language classes.

  If $\cP' \subseteq \cP$ and $\cP$ is at most $k$-fold exponential wrt.\
  $\alpha$, then so is $\cP'$.
\end{thRem}


We say that a formula class $\cF$ is \emph{at most (or at least, respectively)
  $k$-fold exponential wrt.\ $\alpha$} if the class $\{ \Mod{}{\phi} \mid \phi
\in \cF\}$ is. 

The following is the main result of this paper and will be proved in Section~\ref{sec:lower-bound}.
\begin{thThe}\label{th:main}
  Let $k\geq 1$.
  The 
  formula classes
  $\SOSigma 1{\FOPi {k-1}{\sDelta1}}$, $\FO(\sSigma{k})$, and $\sSigma{k}$ are
  both at most and at least $k$-fold exponential wrt.\ group complexity.
\end{thThe}

By Remark~\ref{rem:asymptotic}, this implies.
\begin{thCor}\label{cor:main-sep}
  Let $k\geq 1$.
  There is a picture language definable by a $\sSigma{k+1}$-formula but not by
  an $\FO(\sSigma{k})$-formula.
  In particular, there is a picture language definable by a
  $\FO(\sSigma{k+1})$-formula but not by an $\FO(\sSigma{k})$-formula.
\end{thCor}

The first statement of the above corollary is a new result.
The second statement has already been proved in
\cite{Matz-Diss}, Corollary~2.31.  The proof in this paper is
simpler but requires at least four (or, with standard encoding arguments, at
least two) symbols in the alphabet, whereas the proof in \cite{Matz-Diss}
applies also to singleton alphabets.

Theorem~\ref{th:main} also provides new witnesses for the strictness of the
alternation hierarchy of MSO over pictures, i.e, the result of \cite{Schw97}
that $\sSigma{k+1}$ is more expressive than $\sSigma{k}$.  Unlike the
witnesses in \cite{Schw97, MST98, Matz-Diss}, these witnesses are not over a
singleton alphabet, so the result is formally weaker.


Another consequence of Theorem~\ref{th:main} is the following:
\begin{thCor}\label{cor:AFS00-gen}
  Let $k\geq 1$. There is picture language definable by a
  $\SOSigma 1{\FO(\sSigma 1)}$-formula but not by a $\FO(\sSigma k)$-formula.
\end{thCor}
For the case $k=1$, this is the generalization of
Corollary~\ref{cor:AFS00} to picture languages.

\section{Expressibility Result}
\label{sec:largegroup}

In this section we will prove the upper bound part of Theorem~\ref{th:main}.
For this aim, we will construct a picture language whose height-$m$ fragment
has group complexity $k$-fold exponential in $m$ and that is
definable in the respective formula classes.

The idea is to code large Boolean square matrices with pictures of small height.
The picture language then consists of pictures of the form $p_1 \pcol \ldots
\pcol p_n$, where $p_1,\ldots,p_n$ encode square matrices whose matrix product
over the Boolean semiring is not the zero matrix.
This way, the syntactic monoid contains the monoid of binary relations,
established by the syntactic congruence classes of the encoded matrices.  By a
result of Rhodes (see Theorem~\ref{th:cBn}), that monoid has high group
complexity.

\subsection{Iterated Matrix Multiplication}

Let us consider the semiring of $m{\times}m$-Matrices over the
Boolean semiring.  Its product $\cdot$ is the standard matrix product.
\begin{thLem}\label{lem:matmult}
  Let $m \geq 1$.  Then for every $n\geq 0$ and all matrices $A_1, \ldots, A_n
  \in \{0,1\}^{m\times m}$, we have the following equivalence:
  \[
  \prod_{h=1}^n A_h \neq 0 \Leftrightarrow 
  \exists i_1,\ldots, i_{n+1}\leq m : \forall h \leq n:
  A_h\pos{i_{h}}{i_{h+1}} = 1.
  \]
\end{thLem}
\begin{proof}
  First, we show by induction on $n$ that the following equivalence holds for
  every $n$, every $k,l\in\vonbis1n$ and all matrices $A_1, \ldots, A_n \in
  \{0,1\}^{m\times m}$:
  \[
  \begin{array}{l}
  \left(\prod_{h=1}^n A_h\right)\pos kl = 1 \;\Leftrightarrow\; {}\\
  \qquad\exists i_1,\ldots, i_{n+1}\leq m :
    \begin{array}[t]{@{}l}
      i_1 = k \wedge i_{n+1} = l
      \wedge \big(\forall h \leq n : A_h\pos{i_{h}}{i_{h+1}} = 1\big).
    \end{array}
  \end{array}
  \]
  This equivalence is immediate for $n=0$.  Assume that the equivalence holds
  for some $n\geq 0$, and let $A_1, \ldots, A_{n+1} \in \{0,1\}^{m\times m}$.
  Let $k,l \leq m$. Then $(\prod_{h=1}^{n+1} A_h)\pos kl =
  \big((\prod_{h=1}^n A_h) \cdot A_{n+1}\big)\pos kl$.  Thus the following equivalence
  chain holds: 
  \[
  \begin{array}{rcl}
    \multicolumn{3}{l}{(\prod_{h=1}^{n+1} A_h)\pos kl = 1} \\
    & \Leftrightarrow &  \exists l'\leq m :
    (\prod_{h=1}^n A_h)\pos{k}{l'} = 1 \wedge A_{n+1}\pos{l'}{l} = 1
    \\
    & \Leftrightarrow & \exists i_1,\ldots, i_{n+1},l'\leq m : 
    \begin{array}[t]{@{}l}
      i_1 = k \wedge i_{n+1} = l' 
      \wedge \big(\forall h \leq n : A_h\pos{i_{h}}{i_{h+1}}=1\big) \\
      {} \wedge A_{n+1}\pos{l'}{l}=1
    \end{array}\\
    & \Leftrightarrow & \exists i_1,\ldots, i_{n+1}\leq m : 
    \begin{array}[t]{@{}l}
      i_1 = k
      \wedge \big(\forall h \leq n : A_h\pos{i_{h}}{i_{h+1}}=1\big) \\
      {} \wedge A_{n+1}\pos{i_{n+1}}{l}=1 
    \end{array}\\
    & \Leftrightarrow & \exists i_1,\ldots, i_{n+2}\leq m :
    \begin{array}[t]{@{}l}
      i_1 = k \wedge i_{n+2} = l
      \wedge  \big(\forall h \leq n{+}1 : A_h\pos{i_{h}}{i_{h+1}}=1\big).
    \end{array}\\
  \end{array}
  \]
  This completes the induction.  The lemma is now immediate.
\end{proof}


\subsection{Picture Languages Defined by a Regular Top-Row Language}

For a non-empty picture $p$, we define $\rtop(p)$ to be the word of length
$\length p$ in the top row, i.e., $\rtop(p) = p\pos11\ldots p\pos1{\length p}$.

We observe the following.
\begin{thProp}\label{prop:topreg}
  Let $\Gamma = \schema I$ be an attributed alphabet.
  Let $L\subseteq\Gamma^+$ be a regular word language.  
  Then $\rtop^{-1}(L) := \{ p \in \Gammapp \mid \rtop(p) \in L\}$ is definable by
  a $\nDelta1$-formula.
\end{thProp}
\begin{proof}
  Let $\mA$ be a deterministic finite automaton that recognizes $L$.
  We may assume w.l.o.g.\ that the state set $Q$ of $\mA$ is of the form $\schema
  J$ for some index set $J$ disjoint from $I$.   Let $q_0\in Q$ be the initial
  state of $\mA$ and $F\subseteq Q$ be the set of final states of $\mA$.  Let
  $\Delta \subseteq Q \times \Gamma \times Q$ be the transition relation of $\mA$.

  We will construct a $\nDelta1$-formula that asserts that the uniquely determined
  run of $\mA$ on the top row of a non-empty input picture is accepting.  For this aim, we
  encode the run into an assignment to free variables $(X_\mu)_{\mu \in J}$ in
  the obvious way.

  Recall the formulas $\ftop(x)$, $\fleft(x)$, and $\fright(x)$
  from Example~\ref{exp:ftop}.

  For every $a: I \rightarrow \{0,1\}$, set 
  \[
  \letter_a(x) := 
     \bigwedge_{\mathclap{\mu\in a^{-1}(1)}}      \X\mu x \wedge
     \bigwedge_{\mathclap{\mu\in a^{-1}(0)}} \neg \X\mu x.
  \]
  Similarly, for every $q : J \rightarrow \{0,1\}$, set
  \[
  \State_q(x) :=
     \bigwedge_{\mathclap{\mu\in q^{-1}(1)}}      \X\mu x \wedge
     \bigwedge_{\mathclap{\mu\in q^{-1}(0)}} \neg \X\mu x.
  \]
  Let
  \[
  \phi = 
  \begin{array}[t]{@{}l}
    \forall x 
           \big( \fleft(x) \wedge \ftop(x) \rightarrow \State_{q_0}(x)\big) \\
    {} \wedge \bigwedge_{(q,a,q') \in \Delta} \forall x\forall x' 
    \big(
    \begin{array}[t]{@{}l}
    \State_q(x) \wedge \letter_a(x) \wedge S_2xx' \wedge \ftop(x)
         {} \rightarrow \State_{q'}(x')\big).
    \end{array}
  \end{array}
  \]
  Then $\phi$ asserts, for a non-empty input picture $p$ over alphabet
  $\schema{I\cup J}$, that $\rtop(\restr Jp)$ encodes (in the obvious way) 
  the unique run of $\mA$ on the input word $\rtop(\restr Ip)$.

  Let
  \[
  \psi = \exists x 
       \Big( \ftop(x) 
       \wedge \fright(x)
       \wedge \bigvee\nolimits_{f\in F}\State_f(x)\Big).
  \]
  Then $\psi$ asserts that for a non-empty input picture over alphabet
  $\schema{J}$ that the state encoded in the top right corner is final.
  Consider the formulas 
  \[
  \begin{array}{lcl}
    \sigma & = & \exists (X_\mu)_{\mu \in J} (\phi \wedge \psi), \\
    \pi    & = & \forall (X_\mu)_{\mu \in J} (\phi \rightarrow \psi).
  \end{array}
  \]
  Both $\sigma$ and $\pi$ assert for a non-empty picture $p$ over alphabet $\schema{I}$
  that the unique run of $\mA$ on the top row
  of $p$ reaches a final state, i.e., that its top row is in $L$.
  Thus $\sigma$ and $\pi$ are equivalent and 
  define $\rtop^{-1}(L)$.  This completes the proof.
\end{proof}

\begin{thCor}\label{cor:topin}
  Let $L\subseteq \{0,1\}^+$ be a regular language of non-empty words.
  Let $\Gamma = \schema I$ be some
  attributed alphabet, let $\mu \in I$ be an attribute.
  Define 
  $\rtopin\mu L := \{ p \in \Gammapp \mid \rtop(\attr{p}{\mu}) \in L\}$.
  There exists a $\nDelta1$-formula $\topin \mu L$ such that
  \(
  \Mod{}{\topin\mu L} = \rtopin\mu L.
  \)
\end{thCor}
\begin{proof}
  For every picture $p\in\Gammapp$ we have the equivalence chain
  $p\in\rtopin\mu L$ iff $\rtop(\attr p\mu) \in L$ iff $\attr{\rtop(p)}\mu\in
  L$ iff $p\in \rtop^{-1}(\prinv L\mu)$. Now apply
  Proposition~\ref{prop:topreg} to the regular word language $\prinv L\mu$.
\end{proof}

When we write $\rtopin\mu L$ in the following, the alphabet $\Gamma$ must be
clear from the context.

\subsection{Relativization and Closure Under Concatenation}

In this section, we will show the following result:
\begin{thProp}\label{prop:col-closure}
  For every $k\geq 0$, the class of $\SOSigma 1 {\FOPi k{\nDelta1}}$-definable
  picture languages is closed under column concatenation and column closure.
\end{thProp}
Its proof (see end of this section) is well-known for the important case
$k=0$.
We prepare the proof for the case $k\geq 1$ with the
concept of relativization, cf.\ \cite{Straub94}, Lemma~VI.1.3.

Recall the definition of the $\nDelta1$-formula $x \leq_2 x'$ from
Examples~\ref{exp:leq},~\ref{exp:leqdelta}.
For a non-empty picture $p$ and $j,j'$ with $1\leq j \leq j' \leq \length p$
we define $p[j,j']$ to be the picture that is assembled by the columns $j,
\ldots, j'$ of $p$.

Let $J$, $K$ be disjoint attribute sets.
Let $\phi$ be a formula with free set variables in $\{X_\mu \mid \mu \in J\}$
and free first-order variables in $\{x_\nu \mid \nu \in K\}$.  
Let $x,x'$ be
two fresh first-order variables.  Let $\phi'$ be a formula
whose free variables are those of $\phi$ as well as $x, x'$.  We say that
$\phi'$ \emph{relativizes $\phi$ to $x,x'$} iff for all non-empty pictures
$p\in \Unique {J\cup K}K$ and all $j,j'\leq \length p$ we have
\[
p,\subst{\pos1j}{x}\subst{\pos1{j'}}{x'} \models \phi' \Leftrightarrow
  j \leq j' \mbox{ and } p[j,j'] \models \phi.
\]
Intuitively, $\phi'$ says about the subpicture demarcated by the top-row
positions $x,x'$ the same as $\phi$ says about the whole picture. 

We will need the following obvious remark:
\begin{thRem}\label{rem:interval}
  Let $\phi_{[x,x']}$ relativize $\phi$ to $x,x'$.  Let $p_1,p_2$ be two
  non-empty pictures
  of the same size and $j,j'\leq \length {p_1}$ 
  such that $p_1[j,j'] = p_2[j,j']$.  Then 
  \[
  p_1,\subst{\pos1j}{x}\subst{\pos1{j'}}{x'} \models \phi_{[x,x']} \Leftrightarrow 
  p_2,\subst{\pos1j}{x}\subst{\pos1{j'}}{x' }\models \phi_{[x,x']}.
  \]
\end{thRem}

\begin{thLem}\label{lem:relativizefo}
  For every first-order formula $\phi$ that does not use the first-order variables
  $x, x'$ there is a $\nDelta1$-formula $\phi'$ that relativizes $\phi$ to
  $x,x'$.
\end{thLem}
\begin{proof}
  Set 
  \[
  \columns = \forall x_1\forall x_2 (
                 S_1x_1x_2 \rightarrow (\X\rect x_1 \leftrightarrow \X\rect x_2)).
  \]
  Then $\columns$ asserts that $X_\rect$ is 
  closed under vertical predecessors
  and successors and hence a union of columns.
  Recall the formulas $\ftop(x)$ and $\fleft(x)$ from Example~\ref{exp:ftop}.
  Set
  \[
  \Between = 
  \begin{array}[t]{@{}l}
    \columns \wedge \X\rect x \wedge \X\rect x' \wedge {}\\
    \forall x_1 (S_2x_1x \rightarrow \neg \X\rect x_1)  \wedge {}\\
    \forall x_2 (S_2x'x_2 \rightarrow \neg \X\rect x_2)   \wedge {}\\
    \forall x_1x_2 \big( 
    \begin{array}[t]{@{}l}
      S_2x_1x_2 \wedge \ftop(x_1) \wedge x_1{\not=}x' \wedge x_2{\not=}x 
                 \rightarrow 
                 (\X\rect x_2 \leftrightarrow \X\rect x_1)\big). \\
    \end{array}
  \end{array}
  \]
  Intuitively, $\Between$ asserts for a set $X_\rect$ and two top-row
  positions $x,x'$ that $X_\rect$ is the subblock cyclically
  between the columns marked by the top-row positions $x$ and $x'$.
  %
  In other words, for $j,j',m,n\geq 1$ with $j,j'\leq n$ there exists 
  exactly one picture $p$ of size $(m,n)$ over alphabet $\schema{\{\rect\}}$
  such that $p, \subst{\pos1j}{x}\subst{\pos1{j'}}{x'}\models \Between$.
  That picture is characterized as follows:
  If $j\leq j'$, then $\attr p{\rect}$ carries a $1$ exactly in all positions
  in the columns $j, \ldots, j'$.
  If $j > j'$, then $\attr p{\rect}$ carries a
  $1$ exactly in all positions in the columns 
  $1,\ldots,j',j, \ldots, \length p$.
  
  We introduce a formula $\nowrap$ that asserts that the first
  case is true.
  \[
  \nowrap = \fleft(x) \vee 
          \neg\exists x_2 (\fleft(x_2) \wedge \X\rect x_2).
  \]
  For every non-empty picture $p$ over $\schema{\{\rect\}}$ 
  and every $j,j'\leq \length p$ with
  $p, \subst{\pos1j}{x}\subst{\pos1{j'}}{x'}\models \Between$ we have
  \[
  p, \subst{\pos1j}{x} \models \nowrap \;\Leftrightarrow\; j\leq j'.
  \]

  Let $\phi$ be a first-order formula that does not use the variables $x$,
  $x'$, and $X_\rect$.
  Assume w.l.o.g.\ that there are no universal quantifications in
  $\phi$.  
  Let the formula $\phi'$ result from $\phi$ by relativization to $X_\rect$,
  i.e., by successively replacing every first-order quantification
  of the form $\exists x \psi$ by $\exists x (\X\rect x \wedge \psi)$.

  Define two formulas $\sigma$, $\pi$ as:
  \[
  \begin{array}{rcl}
  \sigma(x,x') & = &
      \exists X_{\rect} ( \Between \wedge \nowrap \wedge \phi'),\\
  \pi(x,x')    & = &
      \forall X_{\rect} ( \Between \rightarrow (\nowrap \wedge \phi')).
  \end{array}
  \]
  Then $\sigma \in \sSigma1$ and $\pi \in \sPi1$.  We have for every
  non-empty picture $p$ and every $j,j' \leq \length p$ that 
  \[
  p, \subst{\pos1j}{x}\subst{\pos1{j'}}{x'}  \models \pi
  \;\Leftrightarrow \;
  (j\leq j' \wedge p[j,j'] \models \phi)
  \;\Leftrightarrow \;
  p, \subst{\pos1j}{x}\subst{\pos1{j'}}{x'} \models \sigma.
  \]
  Hence $\phi_{[x,x']} = \sigma$ is a $\nDelta1$-formula and has the desired
  property. 

  This completes the proof.
\end{proof}

A formula class $\cF$ is \emph{relativizable} iff for every formula $\phi \in
\cF$ and two fresh first-order variables $x,x'$, there is a formula $\phi'\in
\cF$ that relativizes $\phi$ to $x,x'$.

\begin{thLem}\label{lem:relativize}
  For every $k\geq 0$, the formula class $\FOPi k{\nDelta1}$ is relativizable.
\end{thLem}
\begin{proof}
   The proof is by induction on $k$.  For the induction basis, let
   $k=0$ and $\phi \in \sDelta1$.
  Then $\phi$ is equivalent to a formula of the form $\exists
  X_1\ldots X_n (\phi')$ as well as to a formula of the form $\forall
  X_1\ldots X_n (\phi'')$ for two first-order formulas $\phi', \phi''$.
  
  As we have shown in the preceding lemma, there exist a $\sSigma1$-formula
  $\phi'_{[x,x']}$ and a $\sPi1$-formula $\phi''_{[x,x']}$
  that relativize $\phi'$ (and $\phi''$, respectively) to $x,x'$.

  Let 
  \[
  \begin{array}{rcl}
    \sigma & = & \exists X_1\ldots \exists X_n (\phi'_{[x,x']}),\\
    \pi    & = & \forall X_1\ldots \forall X_n (\phi''_{[x,x']}).
  \end{array}
  \]
  Then $\sigma\in\sSigma1$ and $\pi\in\sPi1$.  By
  Remark~\ref{rem:interval}
  they are equivalent and relativize $\phi$ to $x,x'$.
  This completes the proof for the case $k=0$.

  Now assume $k\geq 1$ and the claim is true for $k-1$ instead of $k$.
  Let $\phi\in\FOPi k{\sDelta1}$.
  Then $\phi$ is equivalent to a formula of the form
  \[
    \forall x_1\ldots \forall x_n (\neg \psi),
  \]
  for some $n\geq 1$ and $\psi \in \FOPi {k-1}{\nDelta1}$.
  By assumption there is a $\FOPi {k-1}{\nDelta1}$-formula $\psi_{[x,x']}$ 
  that relativizes $\psi$ to $x,x'$.
  Choose 
  \[
  \phi_{[x,x']} = \forall x_1 \ldots \forall x_n 
        \left(\left(\bigwedge_{i=1}^n x \leq_2 x_i \wedge x_i \leq_2 x'\right) 
                    \rightarrow \neg\psi_{[x,x']}\right).
  \]
  Since $x \leq_2 x'$ is a $\nDelta1$-formula, 
  $\phi_{[x,x']}$ is indeed a $\FOPi {k}{\nDelta1}$-formula.
  Besides, $\phi_{[x,x']}$ relativizes $\phi$ to $x,x'$.  
  This completes the proof.
\end{proof}

In Section~\ref{sec:semantic-fo} we will
come back to the following observation.  See also Remark~\ref{rem:exists-set}.

\begin{thRem}\label{rem:closed-easy}
  Let $\cF$ be a class of formulas closed under conjunction and disjunction.
  The class of $\SOSigma 1{\cF}$-definable picture languages is closed under
  intersection, union, and alphabet projection.
\end{thRem}

\begin{thLem}\label{lem:closed}
  Let $\cF \supseteq \nDelta1$ be a relativizable class of formulas.  Let
  $\Gamma = \schema I$.  Let
  $L_1, L_2 \subseteq \Gammapp$ be two $\SOSigma 1{\cF}$-definable picture
  languages.
  \begin{enumerate}
  \item\label{it:closedpcol} If $\cF$ is closed under conjunction and
    disjunction, then $L_1 \pcol L_2$ is $\SOSigma 1{\cF}$-definable.

  \item\label{it:closedplus} If $\cF$ is closed under conjunction,
    disjunction, and universal $\FO$-quantification, then $L_1^+$ is $\SOSigma
    1{\cF}$-definable.
  \end{enumerate}
\end{thLem}

\begin{proofc}{of Lemma~\ref{lem:closed}, Claim~\ref{it:closedpcol}}
  Let $\phi, \psi \in \SOSigma1\cF$ be formulas with $\Mod{}{\phi} = L_1$ and
  $\Mod{}{\psi} = L_2$.
  Let $x_0,x_1, x_2, x_3$ be four fresh first-order variables. Choose
  formulas $\phi_{[x_0,x_1]}$ and $\psi_{[x_2,x_3]}$ that relativize $\phi$
  and $\psi$, respectively.
  %
  Recall the formulas $\ftop(x)$, $\fleft(x)$, and $\fright(x)$
  from Example~\ref{exp:ftop}.
  Choose 
  \[
  \phi' = \exists x_0,x_1,x_2,x_3 \big(
  \begin{array}[t]{@{}l}
    \ftop(x_0) \wedge 
    \ftop(x_1) \wedge 
    \ftop(x_2) \wedge 
    \ftop(x_3) \wedge {}\\
    \fleft(x_0) \wedge
    S_2x_1x_2 \wedge \fright(x_3) \wedge \phi_{[x_0,x_1]} \wedge \psi_{[x_2,x_3]} \big).
  \end{array}
  \]
  Then $\phi' \in \FOSigma 1{\SOSigma 1{\cF}} = \SOSigma 1{\cF}$.  It remains
  to show that $\Mod{}{\phi'} = L_1 \pcol L_2$.

  Let $p\in \Mod{}{\phi'}$.  Then there exist top-row positions
  $\pos1{j_0}, \pos1{j_1}, \pos1{j_2}, \pos1{j_3} \in \dom p$ such that
  $j_0 = 1$ and $j_1+1 = j_2$ and $j_3=\length p$ and $p[j_0,j_1] \models \phi$ and
  $p[j_2,j_3] \models \psi$.  This implies $p\in \Mod{}\phi \pcol \Mod{}\psi =
  L_1 \pcol L_2$. Thus we have shown $\Mod{}{\phi'} \subseteq L_1 \pcol L_2$. 
  The converse inclusion is similar.
  This completes the proof of
  Claim~\ref{it:closedpcol}.
\end{proofc}

\begin{proofc}{of Lemma~\ref{lem:closed}, Claim~\ref{it:closedplus}}
  Let $\ov X = (X_\mu)_{\mu \in J} $ be a tuple
  of set variables and $\phi \in\cF$ be a formula such that
  \[
  \Mod{}{\exists \ov X \phi} = L_1.
  \]
  Let $\fin$ be a fresh attribute.
  Choose $\nDelta1$-formulas $\topin\fin{0^*1}$ and $\topin\fin{\{0,1\}^*1}$ according to
  Corollary~\ref{cor:topin}.  Set $\phi' = \phi \wedge \topin\fin{0^*1}$.
  Let $x,y$ be fresh first-order variables.  Let $\phi'_{[x,y]}$ be an
  $\cF$-formula that relativizes $\phi'$ to $x,y$.
  Let $\fnobetween$ be a $\sDelta1$-formula that relativizes
  $\topin\fin{0^*1}$ to $x,y$.

  Recall formula $\ftop(x)$ from Example~\ref{exp:ftop}.
  Set
  \[
  \begin{array}{rcl}
  \nextfin & = & 
  \ftop(x) \wedge \ftop(y)
  \wedge \big(\forall w (S_2 wx \rightarrow X_\fin w)\big)
  \wedge \fnobetween.\\
  \phi'' &= &\forall x \forall y
  (\nextfin \rightarrow \phi'_{[x,y]}).\\
  \phi''' &= &\exists \ov X \exists X_\fin (\phi'' \wedge
  \topin\fin{\{0,1\}^*1}).
  \end{array}
  \]
  Then $\nextfin$ is a $\sDelta1$-formula and  $\phi''' \in \SOSigma 1\cF$.
  We claim that 
  \begin{equation}
    \label{eq:L1+}
    \Mod{}{\phi'''} = L_1^+.
  \end{equation}
  Let $I' = I \cup J \cup \{\fin\}$.

  First let $p\in\Mod{I}{\phi'''}$.  There exists a picture
  $p''\in\Mod{I'}{\phi''}$ such that $\exset {J\cup\{\fin\}}{p''} = p$
  and $\rtop(\attr{p''}\fin) \in \{0,1\}^*1$.

  Pick $n\geq 1$ and top-row positions $\pos1{j_1},\ldots, \pos1{j_n}$ such
  that $j_1 < \ldots < j_n$ such that these positions are those top-row
  positions of $p''$ that carry $1$ for attribute $\fin$.  (Note that $j_n =
  \length p$.)  Choose $j_0 = 0$.

  Choose a decomposition $p'' = p'_1 \pcol \ldots \pcol p'_n$ into non-empty
  pictures $p'_1,\ldots,p'_n$ 
  such that $\length{p'_1 \pcol \ldots \pcol p'_k} = j_k$ for every
  $k\in\vonbis1n$.
  For every $k$, choose picture $p_k$ over alphabet $\schema I$ such that
  $\exset {J\cup\{\fin\}}{p'_k} = p_k$.
  Then
  \[
  p = p_1 \pcol \ldots \pcol p_n.
  \]
  %
  Let $k\in\vonbis1n$.  The position $\pos1{j_{k-1}+1}$ is a position in the
  picture $p''$ whose $S_2$-predecessor either does not exist (in case $k=1$)
  or carries a $1$
  for attribute $\fin$.  Thus
  \[
  p'', \subst{\pos1{j_{k-1}+1}}{x}\subst{\pos1{j_k}}{y} \models \nextfin.
  \]
  Since $p''\models \phi''$, this implies 
  \[
  p'', \subst{\pos1{j_{k-1}+1}}{x}\subst{\pos1{j_k}}{y}
  \models \phi'_{[x,y]},
  \]
  which in turn implies
  $p'_k \in \Mod{I'}{\phi'} \subseteq \Mod{I'}{\phi} = \Mod{I \cup J \cup \{\fin\}}{\phi}$.
  Thus $\exset{\{\fin\}}{p'_k} \in \Mod {I\cup J}\phi$.
  By choice of $p_k$, we have $p_k\in\Mod{J}{\exists{\ov X}\phi} = L_1$.
  
  Since $k$ has been chosen
  arbitrarily from $\vonbis1n$, this implies $p = p_1 \pcol \ldots \pcol
  p_n \in L_1^n \subseteq L_1^+$.
  This completes the proof of the direction ``$\subseteq$'' of
  Equation~(\ref{eq:L1+}).

  For the converse direction, let $p\in L_1^+$.  Pick $n\geq 1$ and
  $p_1,\ldots,p_n\in L_1$ such that $p = p_1 \pcol \ldots \pcol p_n$.
  Choose $j_0 = 0$ and for $k\in\vonbis1n$, choose 
  $j_k = \length{p_1 \pcol \ldots \pcol p_k}$.

  For every $k$, we have $p_k \in \Mod{I} {\exists \ov X (\phi)}$, thus for
  every $k$, there exists $p'_k\in \Mod{I'}{\phi}$ such that
  $\exset {J\cup\{\fin\}}{p'_k} = p_k$.
  Furthermore, we may pick $p'_k$ in such a way that
  $\rtop(\attr{p'_k}\fin) \in 0^*1$.

  Choose $p'' = p'_1 \pcol \ldots \pcol p'_n$.  Then
  $\exset {J\cup\{\fin\}}{p''} = p$ and
  $\rtop(\attr{p''}\fin) \in (0^*1)^n \subseteq \{0,1\}^*1$.  
  In order to show $p \in \Mod{I}{\phi'''}$, it remains to show $p'' \in
  \Mod{I'}{\phi''}$.
  To see this, let $z,z'$ be positions in $p''$ such that 
  $p'',\subst{z}{x}\subst{z'}{y} \models \nextfin(x,y)$.
  By definition of $\nextfin$, these are top-row positions, so there are
  $j$, $j'$ such that $z=\pos1j$ and $z' = \pos1{j'}$, and
  (because $\rtop(\attr{p'_k}\fin)\in 0^*1$ for every $k$),
  the columns $j$ and $j'$ are the start- and end-column of one and the
  same $p'_k$-subblock of $p''$ (for some $k\in\vonbis1n$).
  Since $p'_k \models \phi \wedge \topin\fin{0^*1}$, we conclude that
  this $p'_k$-subblock is in $\Mod{}{\phi'}$, i.e.,
  $p'', \subst{z}{x}\subst{z'}{y} \models \phi'_{[x,y]}$.
  Since $z,z'$ have
  been chosen arbitrarily, this implies that $p''\models \phi''$.
  We also have $p''\models \topin\fin{\{0,1\}^*1}$.
  This shows $p\models \phi'''$ and completes the proof of
  Equation~\ref{eq:L1+} and of Lemma~\ref{lem:closed}.
\end{proofc}

We are now ready to prove the result of this section.

\begin{proofc}{of Proposition~\ref{prop:col-closure}}
  For $k \geq 1$, this is a consequence of Lemmas~\ref{lem:relativize},
  \ref{lem:closed} and the fact that $\FOPi k{\nDelta1}$ is closed under
  conjunction, disjunction, and universal first-order quantifications.

  For $k=0$, we have $\SOSigma 1 {\FOPi k{\nDelta1}} = \sSigma1$, which is the
  class of recognizable picture languages.  For this class, the statement is
  well known, see e.g.\ \cite{GRST96}.
\end{proofc}

\subsection{Assembling the Picture Language}

For the rest of this section, let $\cF \supseteq \nDelta1$ be a class of
formulas such that the class of $\SOSigma 1\cF$-definable picture languages
is closed under column
concatenation, column closure, union and intersection.  We will need our
results only for the cases provided by  Proposition~\ref{prop:col-closure},
i.e., for the
case $\cF = \FOPi k{\nDelta1}$ for some $k\geq 0$.

For a function $f : \NNone \rightarrow \NNone$ and an alphabet $\Gamma$, the
picture language \emph{associated to} $f$ is defined as
\[
  \fpics{f,\Gamma} = \{ p \in \Gamma^{+,+} \mid \length p = f(\height p) \}.
\]

Typically $\Gamma$ is  clear from the context or irrelevant because of
Remark~\ref{rem:restr}, so we usually omit the second subscript.

If $f : \NNone \rightarrow \NNone$ is a function, we define another function
$f+1 : \NNone \rightarrow \NNone$, $n \mapsto f(n) + 1$.
\begin{thRem}\label{rem:f+1}
  Let us denote the set of all non-empty pictures of length $1$ over $\Gamma$
  by $\Gamma^{+,1}$.  Then $\fpics{f+1, \Gamma} = \fpics{f,\Gamma} \pcol
  \Gamma^{+,1}$.  Clearly $\Gamma^{+,1}$ is $\cF$-definable, thus if $\fpics
  f$ is $\SOSigma 1\cF$-definable, then, by Lemma~\ref{lem:closed}, so is
  $\fpics{f+1}$.
\end{thRem}

If $n\geq 1$ and $i,j\in\vonbis1n$, set $\code_n(i,j) = (i-1)n + j$.
Then $\code_n$ defines a bijection from $\vonbis1n \times \vonbis1n$ onto
$\vonbis1{n^2}$.

If $n\geq 1$ and $p$ is a picture of length $n^2$ over $\{0,1\}$,
we define $\fold_n(p)$ as the $n{\times}n$-matrix over $\{0,1\}$
with $\fold_n(p)\pos ij = p\pos1{\code_n(i,j)}$ for every $i,j\in\vonbis1n$.
%
%

Every top-row position of a picture $p$ of length $n^2$ corresponds to
exactly one
position in the $n{\times}n$-matrix $\fold_n(p)$.
%
The intuition for the attributes used in the following lemma is:
Picture positions marked by $\blk$ (or $\diag$, or $\fin$) correspond to the
matrix positions in the right column (or on the diagonal, or in the
bottom-right corner, respectively).

\begin{thLem}\label{lem:smalldef}
  Let $f : \NNone \rightarrow \NNone$ be a function such that its associated
  picture language $\fpics f$ is $\SOSigma 1\cF$-definable.
  Consider attribute set $I=\{\diag, \fin, \blk\}$.

  There exists a $\SOSigma 1{\cF}$-definable picture language $L_0$ such that 
  for all pictures $p$ over alphabet $\schema I$ we have: $p\in L_0$ iff
  $\length p = f(\height p)^2$ 
  and for  every $k\leq \length p$ we have the following three
  conditions:
  \begin{eqnarray}
    \label{eq:mod-blk}
    p\pos1k(\blk) =1 & \Leftrightarrow &  f(\height p) \,|\, k,\\
    \label{eq:mod-diag-2}
    p\pos 1k(\diag)  =  1 & \Leftrightarrow &
    (f(\height p)+1)\,|\,(k-1),\\
    \label{eq:fin}
    p\pos1k(\fin) = 1 & \Leftrightarrow & k = f(\height p)^2.
  \end{eqnarray}
\end{thLem}
\begin{proof}
  First, we observe that for every $k\leq \length p$ the
  Equivalence~(\ref{eq:mod-diag-2}) is equivalent to 
  \begin{equation}
    \label{eq:mod-diag}
    p\pos1k(\diag) =1 \Leftrightarrow 
    k \in \{ \code_{f(\height p)}(i,i) \mid i \in \vonbis1{f(\height p)} \}.
  \end{equation}
  This is because  for every $k\leq \length p$, we have the
  equivalence chain:
  $({f(\height p)+1}) \mid ({k-1})$ \tiff{} 
  $\exists a\sgeq 0 : {a({f(\height p)+1})} = {k-1}$ \tiff{}
  $\exists a\sgeq 0 : a f(\height p) + a + 1 = k$ \tiff{}
  $\exists a\sgeq 0 : \code_{f(\height p)}({a+1},{a+1}) = k$ \tiff{}
  $\exists i\sgeq 1 : \code_{f(\height p)}(i,i) = k$.

  Recall the definition of $\rtop^{-1}_{\mu}$ from Corollary~\ref{cor:topin}.
  Let 
  \[
  \begin{array}{rcl}
  M_1 & = &
    \fpics f \cap \rtopin\blk{0^*1},\\
  M_2 & = &
    \fpics {f+1} \cap \rtopin\diag{0^*1},\\
  M_3 & = &
   \rtopin\diag{1} \pcol M_2^+,\\
  M_4 & = & 
    M_1^+  \cap M_3 \cap \rtopin\fin{0^*1}, \\
  X & = & \{ a \in \schema I \mid 
    a(\diag) = 1 \wedge a(\blk) = 1 \leftrightarrow a(\fin) = 1 \},\\
  M_5 & = &
    \rtop^{-1}(X^+),\\
  L_0 & =  & M_4 \cap M_5.
  \end{array}
  \]
  Then $L_0$ is $\SOSigma 1{\cF}$-definable by Corollary~\ref{cor:topin},
  Remarks~\ref{rem:closed-easy}, \ref{rem:f+1}, and Lemma~\ref{lem:closed}.
  
  For the only-if-direction of the lemma, let $p\in L_0$.
  Since $p\in M_1^+$, it may be
  decomposed into subpictures $p = p_1\pcol \ldots \pcol p_n$ such that for
  each $h\leq n$ we have $\length{p_h} = f(\height p)$ and
  $\rtop(\attr{p_h}{\blk}) \in 0^*1$.  This implies (\ref{eq:mod-blk}) for every
  $k\leq \length p$.

  Since $p\in M_3$, it may be decomposed into subpictures $p = q \pcol p'_1
  \pcol \ldots \pcol p'_{n'}$ such that $n'\geq 1$ and $\length q =1$ and
  $\rtop(\attr{q}{\diag}) = 1$ and for every $h\leq n'$ we have 
  $\length{p'_h} = f(\height p) + 1$ and $\rtop(\attr{p_h}{\diag}) \in 0^*1$.
  Thus for every $k\leq \length p$ we have Equivalence~(\ref{eq:mod-diag-2})
  and hence~(\ref{eq:mod-diag}).

  Since $p\in \rtopin\fin{0^+1}$, the position $\pos1{\length p}$ is the only
  top-row position that carries a $1$ for attribute $\fin$, thus (by $p\in
  M_5$), the value $\length p$ is the only value for $k$ that fulfills the
  right sides of both (\ref{eq:mod-blk}) and (\ref{eq:mod-diag-2}), thus
  \[
  \length p = \min\{k > 1 \mid f(\height p)   \,|\, k \wedge 
                                    (f(\height p)+1)\,|\,(k-1)\}.
  \]
  Thus $\length p = f(\height p)^2$ and (\ref{eq:fin}).
  This completes the proof of the only-if-direction.

  For the converse direction, let $p$ be a picture over $\schema I$ such that
  $\length p = f(\height p)^2$ and
  Equivalences~(\ref{eq:mod-blk})-(\ref{eq:fin}) hold for every $k\leq \length
  p$.
  Equivalences~(\ref{eq:mod-blk})-(\ref{eq:fin}) and
  $\length p = f(\height p)^2$ imply $p\in M_5$.
  Equivalence~(\ref{eq:fin}) and $\length p =
  f(\height p)^2$ imply $\rtop(\attr{p}{\fin}) \in 0^*1$.
  From~(\ref{eq:mod-blk})
  we conclude $p\in M_1^+$.  From~(\ref{eq:mod-diag-2}) we conclude $p\in
  M_3$.  This implies $p\in M_1^+ \cap M_3 \cap \rtopin\fin{0^*1} = M_4$.
  This completes the proof of the lemma.
\end{proof}

The intuition that the top-row positions of a picture correspond to square matrix
positions is helpful to understand the next lemma, too. Here, the attributes
intuitively mean:
The attribute $\inp$ marks those picture positions where the Boolean
``input'' matrix carries a one.  The attribute $\piv$ picks one of those as
the ``pivot''.  The attributes $\row$ (and $\col$) mark the position on the
diagonal that shares the row (or column, respectively) with the pivot
position.

\begin{thLem}\label{lem:mediumdef}
  Let $f : \NNone \rightarrow \NNone$ be a function such that its associated
  picture language $\fpics f$ is $\SOSigma 1\cF$-definable.
  There exists a $\SOSigma 1{\cF}$-definable picture language $L_7$ such
  that for all pictures $p$ over alphabet $\schema{\{\inp,\piv,\row,\col,\fin\}}$ we
  have: $p\in L_7$ iff $\length p = f(\height p)^2$ and there are $i,j \leq
  f(\height p)$ such that for all $k \leq \length p$ we have
  \begin{eqnarray}
    \label{eq:medium-fin}
    p\pos1k(\fin) = 1 & \Leftrightarrow & k = f(\height p)^2,\\
    \label{eq:medium-one-inp}
    p\pos1k(\piv) = 1 & \Rightarrow     & p\pos1k(\inp) = 1,\\
    \label{eq:medium-one}
    p\pos1k(\piv) = 1 & \Leftrightarrow & k = \code_{f(\height p)}(i,j), \\
    \label{eq:medium-row}
    p\pos1k(\row) = 1 & \Leftrightarrow & k = \code_{f(\height p)}(i,i) ,\\
    \label{eq:medium-col}
    p\pos1k(\col) = 1 & \Leftrightarrow & k = \code_{f(\height p)}(j,j).
  \end{eqnarray}
\end{thLem}
\begin{proof}
  Recall the definition of $L_{f+1}$ from Remark~\ref{rem:f+1}.
  Choose $L_0$ over the alphabet $\schema{\{\\diag,\fin,\blk\}}$ according to
  Lemma~\ref{lem:smalldef}.
  We define the following picture languages over alphabet $\Gamma =
  \schema{\{\inp,\piv,\row,\col,\fin,\diag,\blk\}}$.
  \[
  \begin{array}{rcl}
    L_1& = & \exsetinv{\{\inp,\piv,\row,\col\}}{L_0}, \\
    L_2 & = & \rtopin\piv{0^*} \pcol 
               (\rtopin\inp1 \cap \rtopin\piv1) \pcol
              \rtopin\piv{0^*}, \\
    L_3 & = & \rtopin\row{0^*} \pcol 
    (\rtopin\diag1 \cap \rtopin\row1)\pcol
    \rtopin\row{0^*}, \\
    L_4 & = & \rtopin\col{0^*}\pcol
    (\rtopin\diag1 \cap \rtopin\col1)\pcol
    \rtopin\col{0^*}, \\
    N_5 & = & \rtopin\blk{0^*1 \cup 0^*} \cap \big(
    \begin{array}[t]{@{}l}
      (\rtopin\row{10^*} \cap \rtopin\piv{0^*1}) \cup {} \\ 
      (\rtopin\piv{10^*} \cap \rtopin\row{0^*1})\big),
    \end{array}\\ 
    L_5 & = & \Gamma^{*,*} \pcol N_5 \pcol \Gamma^{*,*},\\ 
    N_6 & = & \big(\fpics{f}^+\pcol \Gamma^{+,1}\big) \cap \big(
    \begin{array}[t]{@{}l}
      (\rtopin\col{10^*} \cap \rtopin\piv{0^*1}) \cup {} \\ 
      (\rtopin\piv{10^*} \cap \rtopin\col{0^*1})\big),
    \end{array}\\ 
    L_6 & = & \Gamma^{*,*} \pcol N_6 \pcol \Gamma^{*,*}.\\ 
  \end{array}
  \]
  Finally, we define the following picture language over alphabet
  $\schema{\{\inp,\piv,\row,\col,\fin\}}$:
  \[
    L_7 = \exset{\{\diag,\blk\}}{L_1 \cap \ldots \cap L_6}.
  \]
  By Corollary~\ref{cor:topin}, Remark~\ref{rem:closed-easy},
  and Lemma~\ref{lem:closed}, $L_7$
  is $\SOSigma 1{\cF}$-definable.  


  For the only-if-direction, let $p\in L_7$.  For abbreviation, set $n =
  f(\height p)$.
  Pick picture
  $\hat{p}$ over $\Gamma$ such that $\hat p \in L_1 \cap \ldots \cap L_6$ and
  $ \exset{\{\diag,\blk\}}{\hat{p}} = p$.
  We have to show the claim about the size of $p$ as well as the
  Implications~(\ref{eq:medium-fin})-(\ref{eq:medium-col}) for every $k\leq
  \length p$.
  \smallskip

  The choice of $L_1$ (or rather, of $L_0$) implies
  $\length p = n^2$ as well as (\ref{eq:medium-fin}).

  Implication~(\ref{eq:medium-one-inp}) follows from $\hat p \in L_2$.

  Since $\hat p \in L_2$, there exists $k_\piv$ such that 
  \[
  \forall k : p\pos1k(\piv) = 1 \Leftrightarrow k_\piv = k.
  \]
  Choose $i, j \leq n$ such that
  $k_\piv = \code_{n}(i, j)$.  
  This ensures Equivalence~(\ref{eq:medium-one}).

  Similarly, since $\hat p \in L_3$, there exist $k_\row$ such that
  \[
  \forall k : p\pos1k(\row) = 1 \Leftrightarrow k_\row = k.
  \]
  Choose $i_\row, j_\row \leq n$ such that
  \[
    k_\row  = \code_{n}(i_\row, j_\row).
  \]
  Since $\hat p \in L_5$, the substring of $\rtop(\hat p)$ demarcated by
  $k_\piv$ and $k_\row$ is in the same of its $n$ blocks of length
  $n$, thus we have $i = i_\row$.
  Since $\hat p \in L_3$ and (\ref{eq:mod-diag}), we have $i_\row = j_\row$.
  Thus $k_\row = \code_{n}(i,i)$, which proves
  Equivalence~(\ref{eq:medium-row}).

  Similarly, since $\hat p \in L_4$, there exist $k_\col$ such that
  \[
  \forall k : p\pos1k(\col) = 1 \Leftrightarrow k_\col = k.
  \]
  Choose $i_\col, j_\col \leq n$ such that
  \[
    k_\col  = \code_{n}(i_\col, j_\col).
  \]
  Since $\hat p \in L_6$, the difference $|k_\col - k_\piv|$ is a multiple of $n$, thus
  $j = j_\col$.  Since $\hat p \in L_4$ and (\ref{eq:mod-diag}),
  we have $i_\col = j_\col$.
  Thus $k_\col = \code_{n}(j,j)$, which proves
  Equivalence~(\ref{eq:medium-col}).  This completes the proof of the only-if-direction of
  this lemma.
  \medskip

  For the converse direction, let $p$ be a picture over alphabet
  $\schema{\{\inp,\piv,\row,\col,\fin\}}$ such that
  $\length p = f(\height p)^2$, and
  let $i,j \leq f(\height p)$ such that for all $k \leq \length p$ we
  have Implications~(\ref{eq:medium-fin})-(\ref{eq:medium-col}).

  Again, we set $n=f(\height p)$ for abbreviation.
  Choose picture $\hat p$ over alphabet $\Gamma$ such that
  $\restr{\{\inp,\piv,\row,\col,\fin\}}{\hat p} = p$ and Equivalences~(\ref{eq:mod-blk})
  and~(\ref{eq:mod-diag-2}) from Lemma~\ref{lem:smalldef} hold for
  $\restr{\{\diag,\fin,\blk\}}{\hat p}$ instead of $p$.

  We have to verify that $\hat p \in L_1 \cap \ldots \cap L_6$.

  \Ad{$L_1$} The picture $\restr{\{\diag,\fin,\blk\}}{\hat p}$ fulfills
  Lemma~\ref{lem:smalldef}, (\ref{eq:mod-blk})-(\ref{eq:fin}), so
  $\restr{\{\diag,\fin,\blk\}}{\hat p}\in L_0$, thus $\hat p \in
  \restrinv{\{\diag, \fin,\blk\}}{L_0} = \exsetinv{\{\inp,\piv,\row,\col\}}{L_0} =
  L_1$.

  \Ad{$L_2$} Choose a decomposition $p = q \pcol r \pcol s$ with $\length q =
  \code_{n}(i,j) -1$ and $\length r=1$.
  By (\ref{eq:medium-one}) we have
  $1 = p\pos1{\code_{n}(i,j)}(\piv) = r\pos11(\piv)$.
  From (\ref{eq:medium-one-inp}) it follows $1 = r\pos11(\inp)$. Thus 
  $r\in\rtopin\piv{1} \cap \rtopin\inp{1}$.
  Again by (\ref{eq:medium-one}), for all $k<\code_{n}(i,j)$ we have
  $0 = p\pos1k(\piv) = q\pos1k(\piv)$, thus $q\in \rtopin\piv{0^*}$.
  Similarly $s\in \rtopin\piv{0^*}$.

  We have shown $p = q \pcol r \pcol s \in \rtopin\piv{0^*} \pcol
  (\rtopin\piv{1} \cap \rtopin\inp{1}) \pcol \rtopin\piv{0^*} = L_2$.

  \Ad{$L_3$} Similar to $L_2$, but apply (\ref{eq:mod-diag}) instead of
  (\ref{eq:medium-one-inp}), and (\ref{eq:medium-row}) instead of
  (\ref{eq:medium-one}).

  \Ad{$L_4$} Similar to $L_2$, but apply (\ref{eq:mod-diag}) instead of
  (\ref{eq:medium-one-inp}), and (\ref{eq:medium-col}) instead of
  (\ref{eq:medium-one}).

  \Ad{$L_5$} Choose a decomposition $p = q \pcol r \pcol s$ with 
  $\length q = \min\{\code_{n}(i,i), \code_{n}(i,j)\}-1$
  and $\length r =  |\code_{n}(i,i)- \code_{n}(i,j)| + 1$.
  To complete the proof that $p\in L_5$, it suffices to show that $r\in N_5$.

  In case $i\leq j$, the picture $r$ is the infix of $p$ at the columns
  $\code_{n}(i,i),\ldots, \code_{n}(i,j)$ (inclusively).
  None of these numbers (except for maybe the last one) is a multiple of
  $n$, thus (by (\ref{eq:mod-blk})) none of the corresponding
  top-row positions carries a $1$ for attribute $\blk$, thus
  $r\in\rtopin\blk{0^*1 \cup 0^*}$.
  Furthermore,  $r\pos1{\code_{n}(i,i)}(\row) = 1$ by (\ref{eq:medium-row}) and
  $r\pos1{\code_{n}(i,j)}(\piv) = 1$ by (\ref{eq:medium-one}), which implies
  $r\in\rtopin\row{10^*} \cap \rtopin\piv{0^*1}$.
  We have shown $r\in N_5$.
  \newline
  In case $i > j$, one similarly shows $r\in\rtopin\blk{0^*1 \cup 0^*}$
  and $r\in\rtopin\piv{10^*} \cap \rtopin\row{0^*1}$,
  which also implies $r\in N_5$.

  \Ad{$L_6$}  Choose a decomposition $p = q \pcol r \pcol s$ with 
  $\length q = \min\{\code_{n}(j,j), \code_{n}(i,j)\}-1$
  and $\length r =  |\code_{n}(j,j)- \code_{n}(i,j)| + 1$.
  To complete the proof that $p\in L_6$, it suffices to show that $r\in N_6$.
  We have that $|\code_{n}(j,j)- \code_{n}(i,j)|$ is a
  multiple of $n$, thus $\length r-1$ is a multiple of $n = f(\height p)$, hence $r\in
  L_{f}^+\pcol \Gamma^{+,1}$.
  \newline
  In case $i\leq j$, one shows $r\in\rtopin\piv{10^*} \cap
  \rtopin\col{0^*1}$, which implies $r\in N_6$.
  \newline
  In case $i > j$, one shows $r\in\rtopin\col{10^*} \cap
  \rtopin\piv{0^*1}$, which also implies $r\in N_6$.
  
  \medskip We have shown $p\in L_1\cap \ldots\cap L_6$, which completes the proof
  of the lemma.
\end{proof}


In the next lemma, we consider a concatenations of a sequence of pictures,
each of which encodes a Boolean ``input'' square matrix in its top row.  The
marker attributes introduced in the preceding lemmas help to calculate one
position of their matrix multiplication of these input matrices.

\begin{thLem}\label{lem:finaldef}
  Consider the alphabet $\Gamma = \schema{\{\inp,\fin\}}$.
  Let $f : \NNone \rightarrow \NNone$ be a function such that its associated
  picture language $\fpics f$ is $\SOSigma 1\cF$-definable. 
  Let $L$ be the picture language
  over $\Gamma$ containing all pictures of the form 
  $p_1 \pcol \ldots \pcol p_n$
  where $n \geq 1$ and 
  $p_1,\ldots, p_n$ have length $f(\height p)^2$ and 
  $\rtop(\attr{p_h}\fin) \in 0^*1$ (for all $h\in\vonbis1n$) and
  $\prod_{h=1}^n \fold_{f(\height p)}(\attr{p_h}\inp) \neq 0$, where the
  product refers to standard matrix multiplication over the Boolean semiring.

  Then 
  $L$ is $\SOSigma 1{\cF}$-definable.
\end{thLem}

\begin{proof}
  Set $\Sigma = \schema{\{\inp,\piv,\row,\col,\fin\}}$.
  Let $L_7$ be the picture language over $\Sigma$ defined in Lemma~\ref{lem:mediumdef}.
  Let $\Sigma^{+,1}$ denote the set of pictures of length $1$ over $\Sigma$.
  Consider the following picture languages over $\Sigma$:
  \[
  \begin{array}{rcl}
    Q & = &\rtopin\fin{0^*1},\\
    M & = & \fpics{f}^* \pcol \Sigma^{+,1} =
    \{ p \in \Sigmapp \mid (\length p-1) \mbox{ is a multiple of } f(\height p)\},\\
    N & = & \rtopin\col{1\{0,1\}^*} \cap 
            \rtopin\row{\{0,1\}^*1} \cap 
            \rtopin\fin{0^* 1 0^+(1 \cup \epsilon)} \cap 
            M, \\
    L_8 & = & (\Sigma^{*,*} \pcol N \pcol  \Sigma^{*,*}) \cap (Q\pcol Q), \\
    L_9 & = & (L_8^{\,+} \pcol Q) \cup L_8^{\,+}, \\
    L_9' & = & (Q \pcol L_9) \cup Q,\\
    L_{10} & = & (L_7^+ \cap L_9 \cap L_9') \cup L_7.
  \end{array}
  \]
  Let $L_{11} = \exset{\{\piv,\row,\col\}}{L_{10}}$.  Then $L_{11}$ is $\SOSigma
  1{\cF}$-definable by Corollary~\ref{cor:topin},
  Remark~\ref{rem:closed-easy}, Lemma~\ref{lem:closed}, and Lemma~\ref{lem:mediumdef}.
  We claim that $L_{11} = L$.

  First, let $p\in L_{11}$.  We will
  show that $p\in L$.
  There exists $\hat p\in L_{10}$ such that $p =
  \exset{\{\piv,\row,\col\}}{\hat p}$.
  Since $\hat{p}\in L_7^+$, there exist
  $n\geq 1$ and $\hat{p}_1,\ldots, \hat{p}_n \in L_7$ such that $\hat{p} = \hat{p}_1 \pcol
  \ldots \pcol \hat{p}_n$.
  Choose pictures $p_1,\ldots, p_n$ over $\Gamma$ such that for every $h\leq n$ we have $p_h =
  \exset{\{\piv,\row,\col\}}{\hat p_h}$.
  Then $p = p_1 \pcol \ldots \pcol p_n$.
  Since (by (\ref{eq:medium-fin})) we have $L_7 \subseteq Q$,
  we obtain $ \rtop(\attr{\hat{p}_h}\fin) \in 0^*1$ and thus
  $\rtop(\attr{p_h}\fin) \in 0^*1$, as required.

  As a shorthand, we set $m = \height p$.
  By Lemma~\ref{lem:mediumdef}, we have $\length {\hat{p}_h} = f(m)^2$ for
  every $h\leq n$.
  Choose $i_1,\ldots, i_n,j_1,\ldots,j_n \leq f(m)^2$ according to
  Lemma~\ref{lem:mediumdef}, i.e., such that for every $h\leq n$ and every
  $k\leq f(\height p)^2$, we have
  \begin{eqnarray}
    \label{eq:one-inp}
    \hat{p}_h\pos1k(\piv) = 1 & \Rightarrow     & \hat{p}\pos1k(\inp) = 1,\\
    \label{eq:one-single}
    \hat{p}_h\pos1k(\piv) = 1 & \Leftrightarrow & k = \code_{f(m)}(i_h,j_h), \\
    \label{eq:row-single}
    \hat{p}_h\pos1k(\row) = 1 & \Leftrightarrow & k = \code_{f(m)}(i_h,i_h), \\
    \label{eq:col-single}
    \hat{p}_h\pos1k(\col) = 1 & \Leftrightarrow & k = \code_{f(m)}(j_h,j_h).
  \end{eqnarray}
  Furthermore, choose $i_{n+1} = j_n$.
  We claim that
  \begin{equation}
    \label{eq:i=j}
    i_{h+1} = j_h \mbox{ for every } h \leq n.
  \end{equation}
  To see this, let $h \leq n$.  If $h=n$, we have $i_{h+1} = j_h$ by choice of
  $i_{n+1}$.  So we may assume $h < n$ and $n\geq 1$.
  We note that $\hat{p}_h \pcol \hat{p}_{h+1} \in \Sigma^{*,*}
  \pcol N \pcol \Sigma^{*,*}$.  (If $h$ is odd, this follows from $\hat{p}\in
  L_9$; if $h$ is even, this follows from $\hat{p}\in L_9'$.)  This means that
  $\hat{p}_h \pcol \hat{p}_{h+1}$ has an infix $r\in N$.  Let $l,l'\leq 2 f(m)^2$ be the
  start and end positions of that infix $r$ in $\hat{p}_h \pcol \hat{p}_{h+1}$.
  Since $\hat{p}_h, \hat{p}_{h+1} \in L_7$, we have $\length{\hat{p}_h} = \length{\hat{p}_{h+1}} = f(m)^2$.
  Since $\hat{p}_h,\hat{p}_{h+1}\in Q$, the only two top row positions of $\hat{p}_h \pcol
  \hat{p}_{h+1}$ carrying a $1$ for attribute $\fin$ are the rightmost positions of
  $\hat{p}_h$ and of $\hat{p}_{h+1}$.  Since $r\in\rtopin\fin{0^*10^+(1 \cup \epsilon)}$, the
  infix $r$ overlaps the center of $\hat{p}_h\pcol \hat{p}_{h+1}$, i.e., $l\leq f(m)^2$ and $l'>
  f(m)^2$.
  Since $r\in N$, we have 
  \begin{eqnarray}
    \label{eq:col1}
    1 & = & r\pos11(\col) =  (\hat{p}_h \pcol \hat{p}_{h+1})\pos1l(\col) = \hat{p}_h\pos1l(\col)\\
    \label{eq:row1}
    1 & = & r\pos1{\length r}(\row) = (\hat{p}_h \pcol \hat{p}_{h+1})\pos1{l'}(\row) = 
    \hat{p}_{h+1}\pos1{l'-f(m)^2}(\row)
  \end{eqnarray}
  By (\ref{eq:col-single}) and (\ref{eq:col1}), we have
  \begin{equation}
    \label{eq:l}
    l = \code_{f(m)}(j_h,j_h) = (j_h-1)f(m)^2 + j_h.
  \end{equation}
  By (\ref{eq:row-single}) and (\ref{eq:row1}), we have
  \begin{equation}
    \label{eq:l'}
    l'-f(m)^2 = \code_{f(m)}(i_{h+1},i_{h+1}) = (i_{h+1}-1)f(m)^2 + i_{h+1}.
  \end{equation}
  Since $r\in M$ we have 
  \begin{equation}
    \label{eq:distance}
    f(m) \mid l'-l = i_{h+1}f(m)^2 + i_{h+1} - (j_h-1)f(m)^2 - j_h.
  \end{equation}
  We may conclude from (\ref{eq:distance})
  that $f(m) \mid i_{h+1} - j_h$ and hence $i_{h+1} = j_h$.
  This completes the proof of (\ref{eq:i=j}).

  For every $h\in\vonbis1n$, choose $A_h = \fold_{m}(\attr{\hat{p}_h}{\inp})$.
  
  Let $h \leq n$.  By (\ref{eq:one-single}), we have 
  $\hat{p}_h\pos1{\code_{f(m)}(i_h, j_{h})}(\piv) = 1$.
  Thus 
  $A_h\pos{i_h}{i_{h+1}} = 
  A_h\pos{i_h}{j_{h}} = \hat{p}_h\pos1{\code_{f(m)}(i_h,j_{h})}(\inp) = 1$.
  (For the last equality, see (\ref{eq:one-inp}).)
  By Lem\-ma~\ref{lem:matmult}, this implies $\prod_{h=1}^n A_h \neq 0$ and
  hence $p = \exset{\{\piv,\row,\col\}}{\hat{p}} \in L$.
  This completes the proof of $L_{11} \subseteq L$.

  For the converse direction, let $p\in L$.  Again, set $m = \height p$.
  Choose $n\geq 1$ and pictures $p_1,\ldots,p_n$
  of length $f(m)^2$ over $\Gamma$ such that
  $p = p_1 \pcol \ldots \pcol p_n$ and
  $\rtop(\attr{p_h}\fin) \in 0^*1$ (for all $h\in\vonbis1n$) and
  $\prod_{h=1}^n \fold_{f(m)}(\attr{p_h}\inp) \neq 0$.

  We claim that there exists a picture $\hat p\in\Sigmapp$
  with $\exset{\{\piv,\row,\col\}}{\hat p} = p$ and
  $\hat p\in L_{10}$. 
  By Lemma~\ref{lem:matmult}, there exist $i_1,\ldots,i_n \leq f(m)$ such that
  $\fold_{m}(\attr{p_h}{\inp})\pos{i_h}{i_{h+1}} = 1$ for every $h<n$.
  For every $h\leq n$, choose $\hat{p}_h$ as follows:
  \[
  \begin{array}{rcl}
    \restr{\{\inp,\fin\}}{\hat p_h}       & = & p_h, \\
    \hat p_h\pos1k(\piv) = 1  &\Leftrightarrow & k = \code_{f(m)}(i_h,i_{h+1}),\\
    \hat p_h\pos1k(\row) = 1  &\Leftrightarrow & k = \code_{f(m)}(i_h,i_{h}),\\
    \hat p_h\pos1k(\col) = 1  &\Leftrightarrow & k = \code_{f(m)}(i_{h+1},i_{h+1}).
  \end{array}
  \]
  (The letters of $\hat p_h$ at positions $\pos ik$ with $i \neq 1$ are irrelevant
  and may be set arbitrarily.)
  It is straightforward to check that $\hat p_h \in L_7$
  for every $h\leq n$.  We set $\hat p = \hat p_1 \pcol \ldots \pcol \hat
  p_n$.  We have $\hat p \in  L_7^n$.
  
  If $n=1$, we have $\hat p\in L_7 \subseteq L_{10}$.  Thus for completing the
  proof that $\hat p\in L_{10}$, we may assume $n \geq 2$ and show that $\hat
  p\in L_9 \cap L_9'$.

  By definition of $Q$ and $L_9$, for showing $\hat p \in L_9$ it
  suffices to show
  \begin{equation}
    \label{eq:13}
    \hat p_{h} \pcol \hat p_{h+1} \in \Sigma^{*,*} \pcol N \pcol  \Sigma^{*,*}
  \end{equation}
  for every odd $h < n$, whereas for $\hat p \in L_9'$ it suffices to show
  (\ref{eq:13}) for every even $h < n$.  So for completing the proof that
  $\hat p \in L_{10}$ it suffices to show (\ref{eq:13}) for every $h < n$.

  Let $h < n$.  Let $l = \code_{f(m)}(i_{h+1},i_{h+1})$ and $l' = f(m)^2 + l$.
  Then the infix of $\hat{p}_{h} \pcol \hat{p}_{h+1}$ from column $l$ to $l'$ (inclusively)
  is indeed in $N$.
  This completes the proof that $\hat p\in L_{10}$ and thus the proof of the
  lemma.
\end{proof}

\begin{thDef}[Monoid of binary relations]\label{def:Bn}
  Let $n\geq 1$. The monoid $B_n$ is the set of binary relations over
  $\vonbis1n$ together with the usual relation product.
\end{thDef}
This monoid of binary relations has been studied extensively.  We will use it
in the lower bound proof (see Section~\ref{sec:lower-bound}).
Besides, we use it to state the next lemma a little more general than needed.
\begin{thLem}\label{lem:upperbound0}
  Consider the alphabet $\Gamma = \schema{\{\inp,\fin\}}$.
  Let $f : \NNone \rightarrow \NNone$ be a function such that its associated
  picture language $\fpics f$ is $\SOSigma 1\cF$-definable.
  There is a $\SOSigma 1\cF$-definable picture language $L$ over $\Gamma$ such
  that for every $m\geq 1$ the syntactic monoid of $\fix Lm$ contains a
  submonoid isomorphic to $B_{f(m)}$.
\end{thLem}
\begin{proof}
  Let $L$ be defined according to Lemma~\ref{lem:finaldef}.  Then $L$ is
  indeed $\SOSigma 1 {\cF}$-definable.  Let $m\geq 1$.  Set
  $n = f(m)$.
  We claim that the syntactic monoid $\Mon{\fix Lm}$ of the height-$m$
  fragment of
  $L$ contains a submonoid isomorphic to $B_n$.
  If $n=1$, this is trivial, so we may assume $n\geq 2$.
    
  For a relation $\pi \in B_n$, its \emph{characteristic matrix} is the
  $n{\times}n$-matrix $A_\pi$ defined by 
  \[
  A_\pi\pos ij = \left\{
    \begin{array}{cl}
      1 & \mbox{if } (i,j) \in \pi,\\
      0 & \mbox{else.}
    \end{array}
    \right.
  \]
  For a relation $\pi \in B_n$, we define $p_\pi$ as the picture of size $(m,
  n^2)$ over $\Gamma$ such that the top row of $p_\pi$ is
  defined as follows: $\rtop(\attr{p_\pi}\fin) \in 0^*1$ and
  $\fold_{n}(\attr{p_\pi}\inp) = A_\pi$.
  The letters of $p_\pi$ at positions $\pos ij$ with $i\neq 1$ are irrelevant
  and may be chosen arbitrarily.

  For a non-empty picture $p$ of height $m$, let
  $[p]$ denote its congruence class wrt.\ the syntactic congruence of $\fix Lm$.
  We claim that the mapping $\phi: \pi \mapsto [p_\pi]$ is an injective
  homomorphism from $B_n$ into the syntactic semigroup of $\fix Lm$.

  First we show that $\phi$ is a homomorphism.  For this aim, we have to show that
  for every $\pi, \tau \in B_n$, the pictures $p_{\pi\tau}$ and $p_\pi \pcol
  p_\tau$ are syntactically congruent wrt.\ $\fix Lm$.  
  So let $\pi,\tau \in B_n$.  Let $q,r\in\Gamma^{m,*}$.
  We will show that
  \begin{equation}
    \label{eq:congruent}
    q\pcol p_{\pi\tau} \pcol r \in L \Leftrightarrow 
    q\pcol p_{\pi} \pcol p_\tau \pcol r \in L
  \end{equation}
  
  Assume $q\pcol p_{\pi\tau} \pcol r \in L$.
  Each picture of $L$ is a column concatenation of one or more blocks of size
  $(m, n^2)$, each of which is in $\rtopin\fin{0^*1}$, i.e., demarcated by a $1$
  for attribute $\fin$ in its upper right corner.  Since $q\pcol p_{\pi\tau} \pcol
  r \in L$ and $p_{\pi\tau}$ is one of these blocks, the pictures $q$ and $r$
  are assembled by zero or more of these blocks.  In other words,
  there are $s, s'\geq 0$ and pictures $q_1,\ldots,q_s,r_1,\ldots, r_{s'}$ of
  size $(m, n^2)$ such that 
  \[
  q = q_1 \pcol \ldots \pcol q_s \mbox{\quad and\quad}  r = r_1 \pcol \ldots \pcol r_{s'}.
  \]
  Let $\cdot$ denote
  standard matrix multiplication over the Boolean semiring.
  Choose 
  \[
  \begin{array}{rcl}
    B & = & \fold_{n}(\attr{q_1}\inp) \cdots \fold_{n}(\attr{q_s}\inp),\\
    C & = &\fold_{n}(\attr{r_1}\inp) \cdots \fold_{n}(\attr{r_{s'}}\inp).
  \end{array}
  \]
  %
  Since $q\pcol p_{\pi\tau} \pcol r \in L$ and by
  \[
    \label{eq:mat-func}
    \fold_{n}(\attr{p_{\pi\tau}}\inp) = A_{\pi\tau} = A_\pi \cdot A_\tau =
    \fold_{n}(\attr{p_\pi}\inp) \cdot \fold_{n}(\attr{p_\tau}\inp),
  \]
  we have
  \[
  0 \neq B \cdot \fold_{n}(\attr{p_{\pi\tau}}\inp) \cdot C = 
  B \cdot \fold_{n}(\attr{p_\pi}\inp) \cdot \fold_{n}(\attr{p_\tau}\inp) \cdot C,
  \]
  which implies $q\pcol p_{\pi} \pcol p_\tau \pcol r \in L$.
  This completes the proof of the direction ``$\Rightarrow$'' 
  of (\ref{eq:congruent}).  The
  other direction is similar.
  This completes the proof that $\phi$ is a homomorphism, 
  so it remains to show that $\phi$ is injective.

  For every $a,b \in \setbis n$ we define the square matrix $C_{a,b}\in
  \{0,1\}^{n\times n}$ by
  \[
  C_{a,b}\pos ij = 1 \Leftrightarrow (i,j) = (a,b).
  \]
  Recall that $n\geq 2$. Then for every $\pi \in B_n$ we have
  \[
  \{(1,b)\} \pi \{(a,2)\} = 
  \left\{
    \begin{array}{cl}
      \{(1,2) \} & \mbox{ if } (b,a)\in\pi\\
      \emptyset & \mbox{ otherwise,}
    \end{array}
  \right.
  \]
  hence for the characteristic matrix $A_\pi$ we have
  \begin{equation}
    \label{eq:cab}
  C_{1,b} \cdot A_\pi \cdot C_{a,2} = 
  \left\{
    \begin{array}{cl}
      C_{1,2} & \mbox{ if } (b,a) \in \pi,\\
      0      &  \mbox{ otherwise,}
    \end{array}
  \right.
  \end{equation}
  where $0$ denotes the zero matrix of size $n{\times}n$.

  Now let $\pi,\pi' \in B_n$ such that $[p_\pi] = [p_{\pi'}]$.
  Then 
  \[
  \forall q,r \in \Gamma^{m,+} :  q \pcol p_\pi \pcol r \in \fix Lm 
               \Leftrightarrow   q \pcol p_{\pi'} \pcol r \in \fix Lm.
  \]
  This implies (for $q$, $r$ with $C_{1,b} = \fold_{n}(\attr q\inp)$ and $C_{a,2} =
  \fold_{n}(\attr r\inp)$):
  \[
  \forall a,b\in\setbis n: C_{1,b} \cdot A_\pi   \cdot C_{a,2} \neq 0
           \Leftrightarrow C_{1,b} \cdot A_{\pi'} \cdot C_{a,2} \neq 0,
  \]
  which by (\ref{eq:cab}) implies
  \[
  \forall a,b\in\setbis n: (b,a) \not\in \pi \Leftrightarrow (b,a) \not\in\pi'.
  \]
  This implies $\pi=\pi'$.  We have shown that $\phi$ is injective,
  which completes the proof.
\end{proof}

We need the following result:
\begin{thThe}[\cite{Schw97, Matz-Diss}]
  Let $k\geq 1$. There is a $k$-fold exponential function $f$ such that
  $\fpics f$ (over a singleton alphabet) is
  $\SOSigma 1 {\FOPi {k-1}{\nDelta1}}$-definable.
\end{thThe}
The formula construction for the above proof is in
\cite{Schw97}, except for the observation that the inner quantifier
block may be chosen as universal as well as as existential.  For a
construction including this observation, see \cite{Matz-Diss}, Theorem~2.29.

\begin{thLem}\label{lem:upperbound}
  Let $k\geq 1$.  There is a $k$-fold exponential function $f$ and a $\SOSigma
  1 {\FOPi {k-1}{\nDelta1}}$-definable picture language $L$ over alphabet
  $\schema{\{\inp,\fin\}}$ such that for
  every $m\geq 1$, the syntactic monoid of $\fix Lm$ contains a submonoid
  isomorphic to $B_{f(m)}$.
\end{thLem}
\begin{proof}
  By the preceding theorem, there is a $k$-fold exponential function $f$ such
  that $\fpics f$ is definable in the specified formula class.
  The class of picture languages definable in
  $\FOPi {k-1}{\nDelta1}$ is closed under column concatenation and column
  closure by Proposition~\ref{prop:col-closure}.
  The claim follows from Lemma~\ref{lem:upperbound0}.
\end{proof}

For the following theorem from \cite{Rhodes74binary} see also the remarks in
\cite{rhodes2009q}, page~307.
\begin{thThe}[\cite{Rhodes74binary}]\label{th:cBn}
  $\group{B_n} = n-1$ for every $n\geq 1$.
\end{thThe}

Lemma~\ref{lem:upperbound} and Theorem~\ref{th:cBn} imply
\begin{thThe}\label{th:upperbound}
  For every $k\geq 1$, the class of $\SOSigma 1 {\FOPi
    {k-1}{\nDelta1}}$-definable picture languages over
  $\schema{\{\inp,\fin\}}$
  is at least $k$-fold
  exponential wrt.\ group complexity.
\end{thThe}
This is the upper bound part of Theorem~\ref{th:main} and the result of this section.


\section{Non-Expressibility Result}
\label{sec:lower-bound}
\subsection{Some More Notation }

In this section, we write a (monoid or semigroup) homomorphism to the right of
its argument.  Consequently, if $\eta : M \rightarrow N$ is a homomorphism,
the image of $M$ under $\eta$ is denoted $M\eta$, and the pre-image of a subset
$X$ of $N$ is denoted $X\eta^{-1}$.  If $\pi : N \rightarrow T$ is another
homomorphism, then $\eta\pi$ denotes the composition of $\eta$ and $\pi$.

We recall some notions in addition to those in Section~\ref{subs:syncong}.
Let $L$ be a word language over $\Gamma$.
If $M$ is a monoid and $\eta : \Gamma^* \rightarrow M$ is a homomorphism such
that there exists a subset $X\subseteq M$ with $L = X\eta^{-1}$, then we
say that $M$ and $\eta$ \emph{recognize} $L$.  We are only interested in the
case that $M$ is finite.

The syntactic homomorphism $\eta_L: \Gamma^* \rightarrow \Mon L$ maps every
word to its syntactic congruence class.  $\Mon L$ and $\eta_L$ recognize $L$.
Besides, if $\eta$ is a homomorphism that recognizes $L$, then $\eta_L$
factors through $\eta$, meaning that there exists a homomorphism $\theta :
\Gamma^*\eta \rightarrow \Mon L$ such that $\eta \theta = \eta_L$.

\subsection{Transition Monoid of an NFA}
\label{sec:NFA-Mon}

Let $L$ be recognized by a non-deterministic finite automaton (NFA) $\mA$ with
$c$ states.  Using the transition structure of $L$, one can construct the
\emph{transition monoid} $M_\mA$ of $\mA$, which is a monoid with $2^{c^2}$
elements that recognizes $L$.
We sketch this construction from \cite{Pin96}.
Let $Q$ be the state set of $\mA$ and let $\Gamma$ be its alphabet.
Let $M_{\mA} = \schema{Q\times Q}$.  
$M_{\mA}$ forms a monoid, with the multiplication given as follows:
for all $q,q''\in Q$ we set
\[
(A \cdot B)(q,q'') = \sum_{q'\in Q} A(q,q') \cdot B(q',q''),
\]
where sum and product on the right refer to the Boolean semiring.
To every letter
$a\in\Gamma$, we assign the element $\delta_a\in M_\mA$ such that for every $q,q'\in Q$
we have $\delta_a(q,q') = 1$
iff there is a transition from $q$ to $q'$ labeled $a$.

The mapping $\delta : a \mapsto \delta_a$ induces a monoid homomorphism
$\Gamma^* \rightarrow M_{\mA}$.
Then $L$ is the pre-image of $\{(i, f) \in Q\times Q
\mid i \mbox{ is initial and } f \mbox{ is final}\}$ under $\delta$,
thus $M_{\mA}$ and $\delta$ recognize $L$.

\subsection{A Semantic Equivalent to First-Order Quantification}
\label{sec:semantic-fo}

In this section, let $I$ be an attribute set, $K \subseteq I$, and $\Gamma =
\schema I$.  Recall from Remark~\ref{rem:exists-set} how the alphabet
projection $\Exset{\{\mu\}}$ corresponds to the set quantification over
variable $X_\mu$.

Similarly, the syntactic concepts of disjunction and negation correspond to
union and complementation wrt.\ $\Unique IK$, respectively,  in the sense that
\[
\begin{array}{rcl}
  \Mod{I,K}{\phi \vee \psi}  & = & 
  \Mod{I,K}{\phi} \cup 
  \Mod{I,K}{\psi},\\
  \Mod{I,K}{\neg \phi}  & = & 
  \Unique IK \backslash \Mod{I,K}{\phi}.
\end{array}
\]

The next definition and remark present an operation on picture languages that
similarly corresponds to the effect of first-order quantification.

\begin{thDef}\label{def:semexists}
  Let $\mu \in I$.  Let $L$ be a picture language over
  alphabet $\{0,1\}^I$.
  Then $\semexists\mu L$ denotes the set of non-empty pictures $p$ over alphabet
  $\{0,1\}^{I\backslash\{\mu\}}$ for
  which there exists a picture $p'\in \Unique{I}{\{\mu\}} \cap L$ such that 
  $p = \exset{\{\mu\}}{p'}$.
\end{thDef}

\begin{thRem}\label{fossynsem}
  Let $\mu \in K$, let $\phi$ be a formula, and let $x_\mu$ be a first-order
  variable.
  \[
  \Mod{I\backslash\{\mu\},K\backslash\{\mu\}}{\exists x_\mu \phi} =
  \semexists {\mu}{\Mod{I,K}{\phi}}.
  \]
\end{thRem}

\subsection{The Block Product}
\label{sec:block-product}

The \emph{block product} has been introduced
in~\cite{RhoTil89}.    
That block product captures the effect of first-order quantification on the
syntactic monoid in the sense of Lemma~\ref{lem:semexists} below.

Following the presentation of \cite{Straub94}, pp.~61-65, 
we prepare the definition of the block product by introducing bilateral
semidirect products.  

Let $S, T$ be finite monoids, and let us write $S$ additively.  Assume a given
left action as in Definition~\ref{def:left-action}.  The action is
\emph{monoidal} if it additionally satisfies:
\begin{itemize}
\item $1s = s$ for every $s\in S$,
\item $t0 = 0$ for every $t\in T$.
\end{itemize}
\emph{Monoidal right actions} are defined dually.  A left and a
 right action
of $T$ on $S$ are \emph{compatible} if $(ts)t' = t(st')$ for every $t,t'\in
T$ and every $s\in S$.

Given a pair of compatible actions of $S$ on $T$, we define the
\emph{bilateral semidirect product} $S \doublestar T$.  This is the set
$S \times T$ with multiplication given by
\[
(s,t) (s',t') = (st' + ts', tt').
\]
If $M$ and $N$ are monoids and the underlying left and right actions are
monoidal, then the bilateral semidirect product $M \doublestar N$ is indeed a
monoid, see \cite{Straub94}, Proposition~V.4.1.

We remark that by \cite{Straub94}, Example~V.4.b, the semidirect product from
Section~\ref{sec:group-complexity} is a special case of the bilateral
semidirect product with the right action defined by $st = s$.

Closely following \cite{Straub94}, we can now define the block product.
Let $M, N$ be monoids, but this time we will write the products in both
of these monoids multiplicatively.  The set $M^{N\times N}$ of all maps from
$N\times N$ into $M$ forms a monoid under the component-wise product, which we
write additively.  That is, for $F_1, F_2 : N\times N \rightarrow M$ we define
\( F = F_1 + F_2 \), where \( F(n_1,n_2) = F_1(n_1,n_2) \cdot F_2(n_1,n_2) \)
for all $n_1,n_2 \in N$.  Thus $M^{N\times N}$ is isomorphic to the direct
product of $|N|^2$ copies of $M$.  The identity of this monoid is the map that
sends every element of $N\times N$ to $1$.  We define left and right actions
of $N$ on $M^{N\times N}$ by
\[
\begin{array}{rcl}
  (n F)(n_1,n_2) & = & F(n_1n, n_2),\\
  (F n)(n_1,n_2) & = & F(n_1, nn_2).
\end{array}
\]
It is straightforward to verify that these equations define a pair of
compatible left and right actions.  The resulting bilateral semidirect product
if called the \emph{block product} of $M$ and $N$ and is denoted $M\square N$.

Recall that, for a word language $L$, we denote its syntactic monoid by $\Mon {L}$.

The following lemma is an adaption of \cite{Straub94}, Lemma~VI.1.2.

\begin{thLem}\label{lem:semexists}
  Let $I$ be an attribute set, $\Gamma=\{0,1\}^I$, $\mu \in I$, $L \subseteq
  \Gammapp$, $m\geq 1$.  Then $U_1 \square \Mon {\fix Lm}$ recognizes
  $\fix{\semexists\mu L}m$.
More precisely, set $J = I\backslash\{\mu\}$ and $\Sigma =
\{0,1\}^{J}$. Consider the syntactic homomorphisms
\[
\begin{array}{r@{\,:\;}c@{\;\rightarrow\;}l}
  \eta_{\fix Lm}              & \Gamma^{m,*} & \Mon {\fix Lm},\\
  \eta_{\fix{\semexists\mu L}m} & \Sigma^{m,*} & \Mon {\fix{\semexists\mu L}m},
\end{array}
\]
of ${\fix Lm}$ and ${\fix{\semexists\mu L}m}$, respectively.

Let
$\pi : U_1 \square \Mon {\fix Lm} \rightarrow \Mon {\fix Lm}, (F,n) \mapsto n$
be the projection
homomorphism, and let $\sigma : \Sigma \rightarrow \Gamma$ be the alphabet
mapping that maps every $a\in\Sigma$ to the letter $\overline a$ with 
$\ov a \restrict J = a$ and $a(\mu) = 0$. We extend $\sigma$ to 
a homomorphism $\Sigma^{m,*}\rightarrow  \Gamma^{m,*}$ as usual.

There exist homomorphisms 
$\zeta : \Sigma^{m,*}\sigma \rightarrow U_1 \square \Mon {\fix Lm}$ and 
$\tau : U_1 \square \Mon {\fix Lm} \rightarrow \Mon {\fix{\semexists\mu L}m}$ 
such that 
$\zeta\pi  = \eta_{\fix Lm} \restriction (\Sigma^m\sigma)^*$ and 
$\sigma \zeta \tau = \eta_{\fix{\semexists\mu L}m}$.
\input{tikz-semexists}
\end{thLem}
\begin{proof}
  We define the mapping $\otimes : \Sigma \times \{0,1\} \rightarrow \Gamma$ the
  following way:
  \[
  (a \otimes b)(\nu) = \left\{
    \begin{array}{cl}
      a(\nu) & \mbox{if } \nu \neq \mu,\\
      b & \mbox{if } \nu = \mu.
    \end{array}
  \right.
  \]
  (Then $\overline a = a \otimes 0$ for every $a\in\Sigma$.)
  If $p\in \Sigmapp$ and $q\in \{0,1\}^{+,+}$ are pictures of the same size,
  we write $p\otimes q$ for the equally sized picture over $\Gamma$ with
  $(p\otimes q)\pos ij = p\pos ij \otimes q\pos ij$ for every $(i,j)\in \dom p$.

  Choose $T\subseteq \Mon {\fix Lm}$ such that $\fix Lm = T\eta_{\fix Lm}^{-1}$.
  Define 
  \[
  \zeta: \Sigma^m\sigma \rightarrow U_1 \square \Mon {\fix Lm},\quad
  \vect{\ov{a_1}}{\ov{a_m}} \mapsto 
  \left(F, \vect{\ov{a_1}}{\ov{a_m}}\eta_{\fix Lm}\right),
  \]
  where $F$ is defined as follows: for all $n_1, n_2\in \Mon {\fix Lm}$, the component 
  $F(n_1,n_2) \in \{0,1\}$ is $0$ iff 
  there exist $b_1,\ldots,b_{m}\in \{0,1\}$ such
  that $b_1\ldots b_m \in 0^* 1 0^*$ and
  \[
  n_1 \cdot \vect{a_1 \otimes b_1}{a_m \otimes b_m}\eta_{\fix Lm} \cdot n_2
  \in T.
  \]
  We extend $\zeta$ to a monoid homomorphism $\zeta : (\Sigma^m\sigma)^*
  \rightarrow U_1 \square \Mon {\fix Lm}$ as usual.

  Then indeed $\zeta \pi  = \eta_{\fix Lm} \restriction (\Sigma^m\sigma)^*$,
  as claimed in the lemma.
  
  Let $\theta = \sigma \zeta$.  Then $\theta$ is a homomorphism
  $\Sigma^{m,*} \rightarrow U_1 \square \Mon {\fix Lm}$.
  Choose 
  \[
  K := \{ G\in \{0,1\}^{\Mon {\fix Lm}\times \Mon {\fix Lm}} \mid G\pos11 = 0\}
  \times  \Mon {\fix Lm} \subseteq U_1 \square \Mon {\fix Lm}.
  \]
  We show that
  $\fix{\semexists{\mu}L}m = K\theta^{-1}$.

  Let $p\in\Gammapp$ be a picture of size $(m,n)$, say 
  $p = \left(\arr{a_{1,1}}{a_{1, n}}{a_{m,1}}{a_{m,n}}\right)$. 
  Then there are appropriate $F_1, \ldots,
  F_{n} \in \{0,1\}^{\Mon {\fix Lm}\times \Mon {\fix Lm}}$ such that
  \[
  \vect{{a_{1,l}}}{{a_{m,l}}}\theta =
   \vect{\overline{a_{1,l}}}{\overline{a_{m,l}}}\zeta = 
   \left(F_1, \vect{\overline{a_{1,l}}}{\overline{a_{m,l}}}\eta_{\fix Lm}\right)
  \]
  for every $l\leq n$.  We have
  \[
  p\theta = 
     \left(F_1,  \vect{\overline{a_{1,1}}}{\overline{a_{m,1}}}\eta_{\fix Lm} \right)
    \cdots
    \left(F_{n}, \vect{\overline{a_{1,n}}}{\overline{a_{m,n}}}\eta_{\fix Lm} \right).
  \]
  By induction over $n$ one shows that there exists $G\in\{0,1\}^{\Mon {\fix Lm}\times \Mon {\fix Lm}}$ such that
  \[
  p\theta = 
     \left(G, \left(\arr{\overline{a_{1,1}}}{\overline{a_{1, n}}}{\overline{a_{m,1}}}{\overline{a_{m,n}}}\right)\eta_{\fix Lm}
     \right).
  \]
  with
  \[
  G\pos 11 = \prod_{l=1}^{n} F_l
  \left( 
    \left(\arr{\overline{a_{1,1}}}{\overline{a_{1, l-1}}}%
                {\overline{a_{m,1}}}{\overline{a_{m,l-1}}}
         \right)\eta_{\fix Lm},
    \left(\arr{\overline{a_{1,l+1}}}{\overline{a_{1, n}}}%
                {\overline{a_{m,l+1}}}{\overline{a_{m,n}}}
         \right)\eta_{\fix Lm} 
  \right).
  \]
  Thus we have the following equivalence chain:

  $p\theta \in K$
  iff $G\pos 11 = 0$
  iff there is $l \leq n $ 
  such that
  \[
  0 = F_l
  \left(
    \left(\arr{\overline{a_{1,1}}}{\overline{a_{1, l-1}}}%
                {\overline{a_{m,1}}}{\overline{a_{m,l-1}}}
         \right)\eta_{\fix Lm},
    \left(\arr{\overline{a_{1,l+1}}}{\overline{a_{1, n}}}%
                {\overline{a_{m,l+1}}}{\overline{a_{m,n}}}
         \right)\eta_{\fix Lm} 
  \right)
  \]
  iff there is $l \leq n $ and $b_1,\ldots,b_{m} \in \{0,1\}$ 
  such that $b_1\ldots b_m \in 0^* 1 0^*$ and
  \[
    \left(\arr{\overline{a_{1,1}}}{\overline{a_{1, l-1}}}%
                {\overline{a_{m,1}}}{\overline{a_{m,l-1}}}
         \right)\eta_{\fix Lm}
    \cdot
    \vect{a_{1,l} \otimes b_1}{a_{m,1} \otimes b_{m}}\eta_{\fix Lm}
    \cdot
    \left(\arr{\overline{a_{1,l+1}}}{\overline{a_{1, n}}}%
                {\overline{a_{m,l+1}}}{\overline{a_{m,n}}}
         \right)\eta_{\fix Lm}
    \in T
  \]
  iff there exists a picture $q$ of size $(m,n)$ over $\{0,1\}$ such there is exactly one position
  $(k,l)\in \dom q$ with $q\pos kl = 1$ and
  \[
    \left(\arr{a_{1,1}}{a_{1, n}}%
                  {a_{m,1}}{a_{m,n}} \otimes q
         \right)\eta_{\fix Lm}
    \in T
  \]
  iff there exists $q\in\Unique{\{\mu\}}{\{\mu\}}$ of size $(m,n)$ such that
  \[
  p \otimes q \in T\eta_{\fix Lm}^{-1} = L
  \]
  iff $p\in \fix{\semexists \mu L}m$.

  Since $p$ was chosen arbitrary, this implies 
  $K\theta^{-1}= \fix{\semexists \mu L}m$.
  Thus $U_1 \square \Mon {\fix Lm}$ and $\theta$ recognize $\fix{\semexists \mu L}m$.
  
  Since the syntactic homomorphism factors through any other homomorphism that
  recognizes the same language, there exists a homomorphism $\tau$ such that
  $\eta_{\fix{\semexists \mu L}m} = \theta\tau = \sigma \zeta \tau$.
  This completes the proof.
\end{proof}

The above lemma and its proof closely follow \cite{Straub94}, Lemma
VI.1.2. Straubing's lemma is obtained from ours by choosing $m=1$.
Besides, \cite{Straub94} does not introduce the operator $\Semexists\mu$ from
Definition~\ref{def:semexists}, so it must be applied to a formula $\phi$ rather than to its language $L = \Mod{}\phi$.
In \cite{Straub94}, that formula is from the first-order theory with order,
but that condition is of no concern, as the syntax of that formula is
irrelevant for the proof. 

\subsection{Using the Group Complexity}

In this section, we finish the proof of Theorem~\ref{th:main}.
A semigroup that does not contain a non-trivial group is called
\emph{aperiodic}.  A semigroup homomorphism $\phi:S\rightarrow T$ is called
\emph{aperiodic} if for every aperiodic subsemigroup $W$ of $T$, its pre-image
$W\phi^{-1}$ is aperiodic.
Recall the definition of group complexity from
Section~\ref{sec:group-complexity}.
The following theorem is from \cite{Rhodes74}, see e.g. \cite{rhodes2009q}, Theorem~4.9.1.
\begin{thThe}[Fundamental Lemma of Complexity]
  \label{th:fundamental}
  Let $S,T$ be semigroups.  Let $\phi:S\rightarrow T$ be a surjective
  aperiodic homomorphism.   Then $\group S = \group T$.
\end{thThe}

The observation of the next lemma has probably been made before, but I did not
find an explicit statement of it in the literature, so we prove it here.
A similar statement (concerning the Malcev product) can be found in
\cite{rhodes2009q}, Corollary~4.9.4.

\begin{thLem}\label{lem:doublestar-aperiodic}
  Let $S, T$ be semigroups such that $S$ is aperiodic.  Then $\group
  {S\doublestar T} = \group T$.
\end{thLem}
\begin{proof}
  We write $S$ additively and write
  $
  n s = \underbrace{s + \cdots + s}_{n \textrm{ times}}
  $
  to denote, for $n\geq 0$ and $s\in S$, the $n$-fold sum of $s$.
  By Theorem~\ref{th:fundamental} it suffices to show that the projection
  homomorphism $\pi : S \doublestar T \rightarrow T$, $(s,t) \mapsto t$ is
  aperiodic.
  
  Let $W$ be an aperiodic subsemigroup of $T$.
  Let $G$ be a cyclic group contained in $W\pi^{-1}$.
  It suffices to show that $G$ is trivial.

  Choose $g\in G$ such that $G$ is generated by $g$, i.e., $G = \{ g,
  g^2,\ldots, g^{|G|}\}$.  Choose $s\in S$ and $t\in T$ such that $g = (s,t)$.
  We have
  \(
  \{ t,t^2,\ldots, t^{|G|}\} = G\pi
  \)
  is a group contained in $W$, thus it is trivial, which implies $t=t^2$.
   Simple induction shows that for every $n\geq 0$ it holds
   \[
   (s,t)^{n+2} =  \Big(st + n(tst) + ts, t\Big).
   \]
  Since $S$ is aperiodic, there exists $n\geq 1$ such that $n (tst) =
  (n+1) (tst)$.
  Now $(s,t)^{n+2}
  = \big(st + n(tst) + ts, t\big)
  = \big(st + (n+1)(tst) + ts, t\big)
  =(s,t)^{n+3}.$
  Since $(s,t)$ generates the group $G$, this implies that $(s,t)$ is the
  identity of $G$, so $G$ is trivial, which completes the proof.
\end{proof}

\begin{thRem}\label{rem:rec-divide}
  Let $L$ be a word language.  Let $M$ be a monoid that recognizes $L$.
  Then $\group{\Mon L} \leq \group M$.
\end{thRem}
\begin{proof}
  Set $n=\group M$.
  Since $M$ recognizes $L$, the syntactic monoid of $L$ is a homomorphic image
  of $M$ and hence $\Mon L \prec M$.
  Since $M\in \pC_n$ and $\pC_n$ is a pseudovariety, this implies $\Mon L \in
  \pC_n$, thus $\group{\Mon L} \leq n = \group M$.
\end{proof}

\begin{thLem}\label{lem:subgroup}
  Let $I$ be an attribute set, let $\mu \in I$ be an attribute, let $L$ be a
  set of pictures over alphabet $\schema I$, let $m\geq 1$.
  Then 
  $\group{\Mon {\fix{\semexists \mu L}m}} \leq \group{\Mon {\fix Lm}}$.
\end{thLem}
\begin{proof}
  By
  Lemma~\ref{lem:semexists}, the monoid $U_1 \square \Mon {\fix Lm}$
  recognizes $\fix{\semexists \mu L}m$.
  By Remark~\ref{rem:rec-divide},
  \[
  \begin{array}{rcl}
  \group{\Mon {\fix{\semexists \mu L}m}} & \leq &
  \group{U_1 \square \Mon {\fix Lm}} \\
   & = & \group{U_1^{\Mon {\fix Lm} \times \Mon {\fix Lm}} \doublestar \Mon {\fix Lm}}.
  \end{array}
  \]
  Since $U_1$ is
  aperiodic, so is $U_1^{\Mon
    {\fix Lm} \times \Mon {\fix Lm}}$.   Thus the claim follows by
  Lemma~\ref{lem:doublestar-aperiodic}.
\end{proof}



Now we show that the group complexity is not increased
by existential first-order quantification, disjunction, and negation.
More precisely:
\begin{thProp}\label{prop:disjunction}
  Let $m\geq 1$.  For every formula $\phi$, we set
  $\Mon\phi = \Mon{\fix{\Mod{I,K}{\phi}}m}$ for abbreviation, where the
  attribute sets $I$, $K$ are understood as 
  containing the indices of all free variables (or all free first-order variables,
  respectively) of $\phi$.

  Let $\phi,\psi$ be formulas.
  Then
  \begin{eqnarray}
    \label{group-exfo}
    \group{\Mon{\exists x_\nu\phi}} & \leq &
    \group{\Mon{\phi}},\\
    \label{group-vee}
    \group{\Mon{\phi\vee\psi}} & \leq &
    \max\{
    \group{\Mon\phi}, \group{\Mon\psi}\},\\
    \label{group-neg}
    \group{\Mon{\neg \phi}} & = &
    \group{\Mon{\phi}}.
  \end{eqnarray}
\end{thProp}
\begin{proof}
  \Ad{(\ref{group-exfo})}
  This is an immediate consequence of Lemma~\ref{lem:subgroup} and
  Remark~\ref{fossynsem}.

  \Ad{(\ref{group-vee})}
  Let $n = \max\{ \group{\Mon\phi}, \group{\Mon\psi}\}$.
  Then $\Mon\phi, \Mon\psi\in \pC_n$.
  The direct product $\Mon\phi \times \Mon\psi$ recognizes the word language
  $\fix{\Mod{}{\phi\vee\psi}}m$, thus $\Mon{\phi\vee\psi}$ is a homomorphic
  image of that direct product.
  Since $\pC_n$ is a pseudovariety, $\Mon{\phi\vee\psi}\in \pC_n$, i.e.,
  $\group{\Mon{\phi\vee\psi}} \leq n$.  This
  completes the proof of (\ref{group-vee}).

  \Ad{(\ref{group-neg})}
  %
  Let $n=\group{\Mon{\phi}}$
  Since $\Mod{I,K}{\neg \phi} = \Unique IK \backslash
  \Mod{I,K}{\phi}$, the word language $\fix{\Mod{I,K}{\neg \phi}}m$ is
  recognized by the direct product $\Mon {\fix{\Unique IK}m} \times \Mon\phi$,
  thus $\Mon{\neg\phi}$ is a homomorphic image of that direct product.
  Since $\Mon {\fix{\Unique IK}m}, \Mon\phi \in \pC_n$ and $\pC_n$ is a
  pseudovariety, we conclude $\Mon{\neg\phi}\in \pC_n$, i.e.,
  $\group{\Mon{\neg\phi}} \leq n$.  Equality follows by symmetry.
\end{proof}



Define $s:\NN \rightarrow \NN, s(m) = 2^m$.  As usual, $s^0(m) = m$ and
$s^{k+1}(m) = s(s^k(m))$ for every $k$.
The function $s^k$ is $k$-fold exponential.
Fix some $k\geq 1$ for the rest of this section.
Furthermore, assume that $\Gamma = \{0,1\}^I$ is an alphabet
with $I = J \cup K$ for disjoint attribute sets $J, K$.

\begin{thThe}[\cite{MT97}]\label{th:mt97}
  Let $\phi$ be a $\sSigma k$-formula with free set variables in $(X_\mu)_{\mu
    \in J}$ and free first-order variables in $(x_\nu)_{\nu \in K}$.  
  Then there exists $c\geq 1$
  such that for every $m\geq 1$ there exists an NFA with at most 
  $s^k(cm)$ states that recognizes $\fix{\Mod{I,K}\phi}m$.
\end{thThe}
The original theorem only states the above result for the case $K =
\emptyset$.  The present form follows easily by using
\(
 \Mod{I,K}\phi = \Mod{I\cup K, \emptyset}\phi \cap \Unique IK.
\)

Let $n\geq 1$.
The transition monoid (see Section~\ref{sec:NFA-Mon}) of an NFA with $n$
states is a submonoid of the monoid $B_n$ of binary relations of an $n$-set
(see Definition~\ref{def:Bn}).

\begin{thProp}\label{prop:divide-Bn}
  Let $\phi$ be a $\sSigma k$-formula.  There exists $c\geq 1$ such that for
  every $m\geq 1$ the syntactic monoid of
  $\fix{\Mod{}\phi}m$ divides $B_{s^k(cm)}$.
\end{thProp}
\begin{proof}
  Choose $c$ according to Theorem~\ref{th:mt97}.  Let $m\geq 1$.  Set $n =
  s^k(cm)$.  Let $\fA$ be an NFA with at most $n$ states that recognizes
  $\fix{\Mod{}\phi}m$.  Let $M$ be the transition monoid of $\fA$.  Since $M$
  recognizes $\fix{\Mod{}{\phi}}m$, the syntactic monoid of
  $\fix{\Mod{}{\phi}}m$ is a homomorphic image of $M$, thus
  $M(\fix{\Mod{}{\phi}}m) \prec M$.  Besides, $M \prec B_n$.  Since $\prec$ is
  transitive, we have $M(\fix{\Mod{}{\phi}}m) \prec B_n$, i.e., the claim.
\end{proof}

\begin{thProp}\label{prop:semibound}
  $\FO(\sSigma {k})$ is at most $k$-fold exponential wrt.\ group complexity.
\end{thProp}
\begin{proof}
  Let $\cF$ be the class of formulas $\phi$ such that the function
  $m\mapsto\group{\fix{\Mod{}\phi}m}$ is at most $k$-fold exponential.  By
  Proposition~\ref{prop:disjunction}, $\cF$ is closed under existential
  first-order quantification, disjunction, and negation, i.e., $\FO(\cF) =
  \cF$.
  
  By Proposition~\ref{prop:divide-Bn} and Theorem~\ref{th:cBn}, the class
  $\sSigma k$ is at most $k$-fold exponential wrt.\ group complexity, i.e.,
  $\sSigma k \subseteq \cF$.
  Thus $\FO(\sSigma k) \subseteq \FO(\cF) = \cF$, which finishes the proof.
\end{proof}

We are now ready to prove the main result of this paper.
\begin{proofc}{of Theorem~\ref{th:main}.}
  %
  By Remark~\ref{rem:calc}, we have
  $\SOSigma 1 {\FOPi {k-1}{\nDelta1}} \subseteq \SOSigma1{\sPi {k-1}}
  \subseteq \sSigma {k} \subseteq
  \FO(\sSigma {k})$.
  Therefore it suffices to show that
  $\SOSigma 1 {\FOPi {k-1}{\nDelta1}}$ is at least $k$-fold exponential
  whereas $\FO(\sSigma {k})$ is at most $k$-fold exponential wrt.\ group
  complexity.  The first fact is Theorem~\ref{th:upperbound}, the second
  is the preceding proposition.
\end{proofc}

\section{Conclusion}

We continued the work of \cite{Matz-Diss, Matz02} to investigate the
expressive power of first-order quantifications in the context of picture
languages.  We have adapted a lemma by Straubing that analyses the effect of
first-order quantifications in terms of monoid complexity.  We combined this
with the height fragment technique invented in \cite{Gia94, GRST96} and used
in the above papers.  This allowed to deduce a new separation result
(Theorem~\ref{th:main}).  It may be stated informally as: Adding one more set
quantifier alternation gives you expressive power that cannot be captured by
adding any number of first-order quantifier alternations.
\smallskip

At the same time we have found a new sequence of picture languages that
witness the strictness of the quantifier alternation hierarchy of monadic
second-order logic.  Unlike the picture languages in \cite{MT97}, these new witness picture languages are not characterized
by the sizes of their pictures, but rather by the group complexity
required to recognize them.

\subsection{Remarks on the Height Fragment Technique}

The height fragment technique (Remark~\ref{rem:asymptotic}) plays a crucial
role for the separation results for picture language classes defined by
quantifier alternation classes of monadic second-order logic.
Therefore, it may be instructive to summarize the measures
that have been considered so far (more or less explicitly) in the literature
and in this paper.
\begin{itemize}
\item The \emph{state set size} measure assigns to every regular word language
  $L$ the minimal number of states of an NFA that accepts $L$.
\item The \emph{singleton length} measure assigns to every singleton word
  language $\{a^n\}$ over a singleton alphabet $\{a\}$ the length $n$ of its
  only element.
\item The \emph{minimal length} measure assigns to every non-empty word
  language over a singleton alphabet the length of its shortest element.
\item The \emph{group complexity} measure assigns to every non-empty
  word language $L$ the group complexity of its syntactic monoid.
\end{itemize}
The proof that the class of $\sSigma1$-definable picture languages over
$\{0,1\}$ is not closed under complement was done in \cite{GRST96} and uses
the state set size measure.  That picture language class is at most
singly exponential wrt.\ state set size, but the class of $\sPi1$-definable
picture languages is not, as it contains a language with state set size
$2^{\Omega(m^2)}$, namely the picture language of all pictures of the form
$p\pcol p$, where $p$ is a non-empty square picture.

The first result involving the singleton length measure is from \cite{Gia94}
and says that the class of recognizable (or, by \cite{GRST96} equivalently,
of $\sSigma1$-definable) picture languages over a singleton
alphabet is both at most and at least $1$-fold exponential.

Generalizing Giammarresi's result, 
in \cite{MT97} (and \cite{Schw97}, respectively) it is shown that the class of
$\sSigma k$-definable picture languages over a singleton alphabet is at most (and
at least, respectively) $k$-fold exponential wrt.\ singleton length.

In \cite{Matz-Diss}, Corollary~3.66 and Theorem~4.25, it is shown that
$\FO(\sSigma {k})$ is both at most and at least $(k{+}1)$-fold
exponential wrt.\ singleton length.

In \cite{Matz-Diss}, Theorem~3.61 and Corollary~4.15, it is shown that $\sPi
k$ is at least $(k{+}1)$-fold exponential whereas $\sSigma k$ is at most
$k$-fold exponential wrt.\
minimal length.  This allowed to separate these two classes for the case of a
singleton alphabet.


\smallskip

Our result Corollary~\ref{cor:main-sep} is based on the group complexity
measure.  That corollary can be proved neither by state set size nor by
singleton length, since for every $k\geq 1$, even the class of
$\FO(\sSigma{1})$-definable picture languages is at least $k$-fold exponential
wrt.\ state set size (\cite{MT97}), and both $\FO(\sSigma{k})$ and
$\sSigma{k+1}$ are both at most and at least $(k{+}1)$-fold exponential wrt.\
singleton length (\cite{Matz-Diss}).

\subsection{Open Questions}

Corollary~\ref{cor:main-sep} states that there is a $\sSigma{k+1}$-definable
picture language that is not $\FO(\sSigma{k})$-definable.
Lemma~\ref{lem:finaldef} shows that the alphabet is $\schema{\{\inp,\fin\}}$,
i.e., the alphabet size is four.  We note
without proof that one can reduce the size of the alphabet to two by applying
standard encoding techniques.
It remains open whether we can even reduce the size to one.  For that case, we
only know:
\begin{thThe}[\cite{Matz-Diss}, Theorems~2.29, 2.30]
  Let $k \geq 1$.
  For a singleton alphabet, there is a 
  $\sSigma{k+1}$-definable picture language that is not
  $\FO(\sSigma{k-1})$-definable.
\end{thThe}

The following problem from \cite{Matz-Diss, JM01} remains open, too.

\begin{thProb}\label{prob:mso}
  Is there is an MSO-formula that is not equivalent to
  any $\SOSigma 1{\FO({\sSigma1})}$-formula?
\end{thProb}

Separation results such as Corollary~\ref{cor:main-sep} may be transferred to
other classes of structures, for example, to directed graphs, using standard
encoding techniques.  This has been carried out formally in \cite{Matz-Diss},
Chapter~5.

%

Problem~\ref{prob:mso} is also open in that more general
setting of directed graphs.  Answering that question might be the
first step towards attacking the \emph{closed hierarchy} as defined in
\cite{AFS00}.

I find the following questions interesting from a methodological point of
view:
\begin{thProb}\label{prob:technique}
  Is there any separation result concerning quantifier alternation classes of
  MSO-formulas (as defined in
  Section~\ref{sec:quantifier-alternation-classes}) that either holds for
  directed graphs but not for pictures, or holds for pictures but
  cannot be proved with the height fragment technique?
\end{thProb}

\section{Acknowledgments}

Thanks go to several participants of the 2010 workshop on ``Circuits, logic,
and games'' in Dagstuhl: Wolfgang Thomas for discussions on the draft of this
paper; Nicole Schweikardt for her detailed analysis and helpful feedback;
Jean-Eric Pin for the copy of \cite{Pin86} and for pointing me helpful
hints 
when I was stuck in the proof of Proposition~\ref{prop:semibound}---later I
discovered that the group complexity was more suitable than the investigation
of the symmetric subgroups, so the respective papers are not cited;
Klaus-J\"orn Lange for a proof related to the block product, even though that
proof did not make it into the present paper because I discovered later that I
had asked the wrong question;
Pierre McKenzie for pointing me to the Landau function, which I did not need
in the end either, because the lengths of the cyclic subgroups then seemed to
be the wrong trace;
and Thomas Colcombet and Etienne Grandjean for discussions on related
subjects.  \smallskip

Last, but not least, I am grateful to Thomas Wilke.  When he pointed me to
Chapter VI in \cite{Straub94}, it was not the first time that he gave me the
crucial hint that I needed to turn a vague idea into a theorem.

{\footnotesize
\bibliography{literatur}}

\end{document}

%% file: defs.tex
\newtheorem{thDef}{Definition}[section]
\newtheorem{thRem}[thDef]{Remark}
\newtheorem{thLem}[thDef]{Lemma}
\newtheorem{thExp}[thDef]{Example}
\newenvironment{Exp}{\begin{thExp}\rm}{\mbox{}\hfill$\Box$\end{thExp}}
\newtheorem{thCor}[thDef]{Corollary}
\newtheorem{thThe}[thDef]{Theorem}
\newtheorem{thProp}[thDef]{Proposition}
\newtheorem{thProb}[thDef]{Problem}

\newenvironment{proof}{\paragraph{Proof}}{\mbox{}\hspace*{\fill}$\Box$\medskip}
\newenvironment{proofc}[1]{\paragraph{Proof \textrm{#1}}}{\mbox{}\hfill$\Box$\medskip}
\newcommand{\length}[1]{|{#1}|}
\newcommand{\height}[1]{\underline{\overline{#1}}}
\newcommand{\Sigmapp}{\Sigma^{+,+}}
\newcommand{\Gammapp}{\Gamma^{+,+}}
\newcommand{\dom}[1]{\mathrm{dom}({#1})}
\newcommand{\mA}{\mathfrak{A}}

\newcommand{\cF}{\mathcal{F}}
\newcommand{\NN}{\mathbb{N}}
\renewcommand{\phi}{\varphi}
\renewcommand{\epsilon}{\varepsilon}
\newcommand{\restrict}[1]{\restriction{#1}}

\newcommand{\fix}[2]{{#1}[{#2}]} 
\newcommand{\vect}[2]{\clmn{ {#1}
    \\[-0.8ex] \vdots \\[-1ex] {#2}}} 
\newcommand{\arr}[4]{
  \begin{array}{c@{}c@{}c}
    {#1}   & \cdots & {#2} \\[-0.8ex]
    \vdots   &        & \vdots\\[-1.3ex]
    {#3}   & \cdots & {#4} 
  \end{array}
  }

 \newcommand{\clmn}[1]{\left(\begin{array}{@{}c@{}} #1
     \end{array}\right)}

\newcommand{\setbis}[1]{\{1,.., {#1}\}}

\newcommand{\vonbis}[2]{\{{#1},..,{#2}\}}
\newcommand{\ganzrec}[2]{\setbis{#1}{\times}\setbis{#2}}
\newcommand{\NNone}{\NN_{\geq 1}}
\newcommand{\cP}{\mathcal{P}}

\newcommand{\fA}{\mathfrak{A}}

\newcommand{\FO}{\textsl{FO}}

\let\ov\overline
\newcommand{\pos}[2]{\langle #1,#2 \rangle} 
\newcommand{\code}{\textsl{code}}

\newcommand{\class}[1]{\textsl{#1}}
\newcommand{\yB}{{\class{B}}}
\newcommand{\yFO}{{\class{FO}}}
\newcommand{\yco}{\class{co}}
\newcommand{\FOSigma}[2]{\Sigma_{#1}^0(#2)}
\newcommand{\FOPi}[2]{\Pi_{#1}^0(#2)}
\newcommand{\SOSigma}[3][{\mathsf{mon}}]{\sSigma[#1]{#2}(#3)}
\newcommand{\SOPi}[3][{\mathsf{mon}}]{\sPi[#1]{#2}(#3)}

\newcommand{\clB}[1]{{\yB({#1})}}
\newcommand{\clFO}[1]{{\yFO({#1})}}
\newcommand{\co}[1]{\yco\textit-{#1}}
\newcommand{\nDelta}[2][{\mathsf{mon}}]{\sDelta[#1]{#2}}

\newcommand{\sSigma}[2][{\mathsf{mon}}]{\Sigma^{#1}_{#2}}
\newcommand{\sPi}[2][{\mathsf{mon}}]{\Pi^{#1}_{#2}}
\newcommand{\sDelta}[2][{\mathsf{mon}}]{\Delta^{#1}_{#2}}
\newcommand{\Mod}[2]{{\textsl{Mod}_{#1}(#2)}}

\newcommand{\schema}[1]{\{0,1\}^{#1}}
\newcommand{\attr}[2]{\mathsf{pr}_{#2}(#1)}
\newcommand{\prinv}[2]{\mathsf{pr}^{-1}_{#2}(#1)}

\newcommand{\Restr}[1]{\mathsf{restr}_{#1}}
\newcommand{\restr}[2]{\Restr{#1}(#2)}
\newcommand{\restrinv}[2]{\mathsf{restr}^{-1}_{#1}(#2)}
\newcommand{\Exset}[1]{\mathsf{ex}_{#1}}
\newcommand{\exset}[2]{\Exset{#1}(#2)}
\newcommand{\exsetinv}[2]{\Exset{#1}^{-1}(#2)}

\newcommand{\fpics}[1]{L_{#1}}

\let\form\textit

\newcommand{\Semexists}[1]{\mathsf{exfo}_{#1}}
\newcommand{\semexists}[2]{\Semexists{#1}(#2)}
\newcommand{\Unique}[2]{\textsl{Unique}_{#1,#2}}
\newcommand{\fold}{\mathsf{fold}}
\newcommand{\rtop}{\mathsf{top}}

\newcommand{\cld}{{\textsl{cld}}}
\newcommand{\inp}{{\textsl{inp}}}

\newcommand{\fin}{{\textsl{end}}}
\newcommand{\rect}{{\textsl{rect}}}

\newcommand{\diag}{{\textsl{diag}}}
\newcommand{\blk}{{\textsl{blk}}}
\newcommand{\row}{{\textsl{row}}}
\newcommand{\col}{{\textsl{col}}}
\newcommand{\piv}{{\textsl{piv}}}

\newcommand{\ftop}{\form{top}}
\newcommand{\fleft}{\form{left}}
\newcommand{\fright}{\form{right}}
\newcommand{\fnobetween}{\form{no-between}}

\newcommand{\nextfin}{\form{next}_\fin}

\newcommand{\letter}{\form{letter}}
\newcommand{\Between}{\form{between}}
\newcommand{\columns}{\form{columns}}
\newcommand{\nowrap}{\form{nowrap}}

\newcommand{\State}{\form{state}}
\newcommand{\topin}[2]{\form{top-in}_{#1}({#2})}
\newcommand{\rtopin}[2]{\rtop^{-1}_{#1}({#2})}

\newcommand{\X}[2]{X_{#1\,}{#2}}

\newcommand{\doublestar}{{\,**\,}}

\newcommand*{\bigset}[1]{\ensuremath{ \left\{ #1 \right\} } }
\newcommand{\ex}{\mbox{\scriptsize{$\exists$}}}
\newcommand{\fa}{\mbox{\scriptsize{$\forall$}}}

\newcommand{\EX}{\mbox{\large{$\exists$}}}
\newcommand{\FA}{\mbox{\large{$\forall$}}}

\newcommand{\blockprod}[1]{}
\newcommand{\Mon}[1]{M({#1})}

\newcommand{\subst}[2]{\tfrac{#1}{#2}}

\newcommand{\Group}{\mathsf{c}}
\newcommand{\group}[1]{\Group({#1})}

\newcommand{\pV}{\mathbf V}
\newcommand{\pW}{\mathbf W}
\newcommand{\pG}{\mathbf G}
\newcommand{\pA}{\mathbf A}
\newcommand{\pC}{\mathbf C}

\newcommand{\Ad}[1]{\smallskip \textbf{Ad #1:\;}}
\DeclareMathSymbol{\square}{\mathbin}{AMSa}{"03}

\newcommand{\tiff}{%
  \hspace{0.1em plus 0.2em minus 0.1em}iff\hspace{0.1em plus 0.2em minus 0.1em}}
\newcommand{\sgeq}{%
  \hspace{0.1em plus 0.1em minus 0.05em}{\geq}\hspace{0.1em plus 0.1em minus 0.05em}}

\renewcommand{\itemhook}{%
  \setlength{\topsep}{5pt plus 5pt minus 2pt}
  \setlength{\itemsep}{0pt plus 5pt}}

\renewcommand{\enumhook}{%
  \setlength{\topsep}{5pt plus 5pt minus 2pt}
  \setlength{\itemsep}{0pt plus 5pt}}
\newcommand{\pcol}{}

%% file: tikz-hierarchy.tex
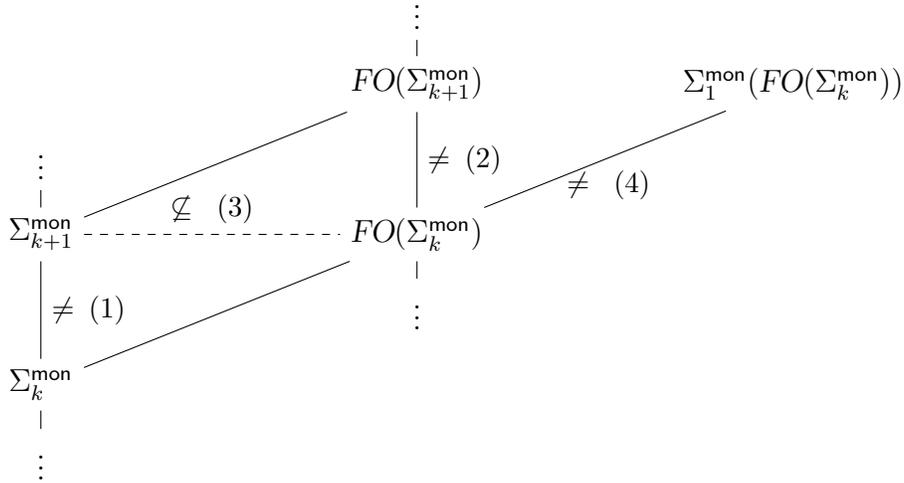
\begin{figure}[tb]
  \centering

\vspace{0cm minus 0.7cm}
\begin{tikzpicture}[node distance=2cm, auto]
  \node (Sigma1)                   {$\sSigma{k+1}$};
  \node (Sigma0) [below of=Sigma1] {$\sSigma{k}$};
  \node (FO0)    [right of=Sigma1, node distance=5cm] {$\FO(\sSigma {k})$};
  \node (FO1)    [above of=FO0]    {$\FO(\sSigma {k+1})$};

  \node (SFO0)   [right of=FO1, node distance=5cm] {$\SOSigma1{\FO(\sSigma {k})}$};

  \node (Sigma2) [above of=Sigma1, node distance=1cm] {$\vdots$};
  \node (Sigma-) [below of=Sigma0, node distance=1cm] {$\vdots$};
  \node (FO-)    [below of=FO0,    node distance=1cm] {$\vdots$};
  \node (FO2)    [above of=FO1,    node distance=1cm] {$\vdots$};

  \draw[-]  (Sigma2) to node                {} (Sigma1);
  \draw[-]  (Sigma1) to node[right] {\small $\not=$\, (1)} (Sigma0);
  \draw[-]  (Sigma0) to node       {} (Sigma-);
  \draw[-]  (FO2)    to node       {} (FO1);
  \draw[-]  (FO1)    to node[right]  {\small $\not=$\, (2)} (FO0);
  \draw[-]  (FO0)    to node       {} (FO-);
  \draw[-]  (Sigma1) to node       {} (FO1);
  \draw[-]  (Sigma0) to node       {} (FO0);

  \path (Sigma1) edge[dashed] 
             node[anchor=south,above] {\small $\not\subseteq$ \, (3)}
        (FO0);
  \draw[-]  (FO0)    to node  [below] {\small \;$\not=$ \, (4)} (SFO0);
\end{tikzpicture}
\vspace{0cm minus 0.7cm}

  \caption{Hasse diagram}
   \begin{minipage}{0.9\hsize}\small\smallskip
     Hasse diagram of the monadic second-order alternation hierarchy and the
     first-order closures hierarchy over pictures
     ($k \geq 1$).  All inclusions are trivial.  The dashed
     line indicates that the two classes are incomparable.

     The non-inclusions~(1) (from \cite{Schw97}) and (2) (from
     \cite{Matz-Diss}) are re-proved here (for non-trivial alphabets), see
     Corollary~\ref{cor:main-sep}.  That corollary shows also the
     non-inclusion~(3), which is new.  Non-inclusion~(4) has been shown for
     directed graphs in \cite{AFS00, JM01} and is here re-proved (for pictures
     over a non-trivial alphabet) as Corollary~\ref{cor:AFS00-gen}.
   \end{minipage}
   \vspace{0cm plus 0.5cm}
  \label{fig:hierarchy}
\end{figure}


%% file: tikz-semexists.tex
\begin{figure}[tb]
  \centering

\begin{tikzpicture}[node distance=1.8cm, auto]
  \node (Sigma) {$\Sigma^{m,*}$};
  \node (Sigmasigma) [below of=Sigma, node distance=3.1cm] {$\underbrace{\Sigma^{m,*}\sigma}$};
  \node (Gamma) [below of=Sigma, node distance=3.6cm] {$\subseteq\Gamma^{m,*}$};
  \node (MonL)  [right of=Gamma, node distance=4.5cm] {$\Mon{\fix Lm}$};
  \node (MonEL) [right of=Sigma, node distance=4.5cm] {$\Mon{\fix {\semexists \mu L}m}$};
  \node (box)   [above of=MonL]                     {$U_1 \square \Mon {\fix Lm}$};
 
  \draw[->] (Sigma) to node [swap] {$\sigma$} (Sigmasigma);
  \draw[->] (Gamma) to node [swap] {$\eta_{\fix Lm}$} (MonL);
  \draw[->] (Sigma) to node [swap] {$\eta_{\fix {\semexists \mu L}m}$} (MonEL);
  \draw[->] (Sigmasigma) to node [swap] {$\zeta$} (box);
  \draw[->] (box)   to node [swap] {$\pi$} (MonL);
  \draw[->] (box)   to node        {$\tau$} (MonEL);
\end{tikzpicture}

  \caption{Commutative Diagram for Lemma~\ref{lem:semexists}}
  \label{fig:semexists}
\end{figure}


%% file: fomon.bbl
\begin{thebibliography}{GRST96}

\bibitem[AFS00]{AFS00}
M.~Ajtai, R.~Fagin, and L.~Stockmeyer.
\newblock The closure of monadic {NP}.
\newblock {\em Journal of Computer and System Sciences}, 60(3):660--716, 2000.
\newblock Journal version of STOC'98 paper.

\bibitem[Eil76]{Eil76}
S.~Eilenberg.
\newblock {\em Automata, Languages, and Machines}.
\newblock Number Bd. 2 in Pure and applied mathematics. Academic Press, 1976.

\bibitem[Gia94]{Gia94}
D.~Giammarresi.
\newblock Two-dimensional languages and recognizable functions.
\newblock In G.~Rozenberg and A.~Salomaa, editors, {\em Developments in
  Language Theory, Proceedings of the conference, Turku (Finnland) '93}, pages
  290--301. world scientific, Singapore, 1994.

\bibitem[GRST96]{GRST96}
D.~Giammarresi, A.~Restivo, S.~Seibert, and W.~Thomas.
\newblock Monadic second-order logic and recognizability by tiling systems.
\newblock {\em Information and Computation}, 125(1):32--45, 1996.
\newblock Journal version of STACS'94 paper.

\bibitem[JM01]{JM01}
D.~Janin and J.~Marcinkowski.
\newblock A toolkit for first order extensions of monadic games.
\newblock In {\em Proceedings of 18th Annual Symposium on Theoretical Aspects
  of Computer Science (STACS'01)}, volume 2010 of {\em Lecture Notes in
  Computer Science}, pages 353--364. Springer-Verlag, 2001.

\bibitem[KR65]{KroRho65}
K.~Krohn and J.~Rhodes.
\newblock Algebraic theory of machines. {I}. {P}rime decomposition theorem for
  finite semigroups and machines.
\newblock {\em Trans. Amer. Math. Soc.}, 116:450--464, 1965.

\bibitem[Mat99]{Matz-Diss}
O.~Matz.
\newblock {\em Dot-depth and monadic quantifier alternation over pictures}.
\newblock PhD thesis, RWTH Aachen, 1999.

\bibitem[Mat02]{Matz02}
O.~Matz.
\newblock Dot-depth, monadic quantifier alternation, and first-order closure
  over grids and pictures.
\newblock {\em Theor. Comput. Sci.}, 270(1-2):1--70, 2002.

\bibitem[MST02]{MST98}
O.~Matz, N.~Schweikardt, and W.~Thomas.
\newblock The monadic quantifier alternation hierarchy over grids and graphs.
\newblock {\em Information and Computation}, 179(2):356--383, 2002.

\bibitem[MT97]{MT97}
O.~Matz and W.~Thomas.
\newblock The monadic quantifier alternation hierarchy over graphs is infinite.
\newblock In {\em Twelfth Annual IEEE Symposium on Logic in Computer Science},
  pages 236--244, Warsaw, Poland, 1997. IEEE.

\bibitem[Pin86]{Pin86}
J.-E. Pin.
\newblock {\em Varieties Of Formal Languages}.
\newblock Plenum Publishing Co., 1986.

\bibitem[Pin96]{Pin96}
J.-E. Pin.
\newblock Logic, semigroups and automata on words.
\newblock {\em Annals of Mathematics and Artificial Intelligence}, 16:343--384,
  1996.

\bibitem[Rho74a]{Rhodes74binary}
J.~Rhodes.
\newblock Finite binary relations have no more complexity than finite
  functions.
\newblock {\em Semigroup Forum}, 7:92--103, 1974.

\bibitem[Rho74b]{Rhodes74}
J.~Rhodes.
\newblock Proof of the fundamental lemma of complexity (strong version) for
  arbitrary finite semigroups.
\newblock {\em J. Comb. Theory, Ser. A}, 16(2):209--214, 1974.

\bibitem[RS09]{rhodes2009q}
J.~Rhodes and B.~Steinberg.
\newblock {\em The q-theory of finite semigroups}.
\newblock Springer monographs in mathematics. Springer, 2009.

\bibitem[RT89]{RhoTil89}
J.~Rhodes and B.~Tilson.
\newblock The kernel of monoid morphisms.
\newblock {\em Information and Computation}, 62:227--268, 1989.

\bibitem[Sch97]{Schw97}
N.~Schweikardt.
\newblock The monadic quantifier alternation hierarchy over grids and pictures.
\newblock In Mogens Nielson and Wolfgang Thomas, editors, {\em Computer Science
  Logic}, volume 1414 of {\em Lecture Notes in Computer Science}, pages
  441--460. Springer, 1997.

\bibitem[Str94]{Straub94}
H.~Straubing.
\newblock {\em Finite automata, formal logic, and circuit complexity}.
\newblock Birkh\"auser Verlag, Basel, Switzerland, 1994.

\end{thebibliography}
